\documentclass[12pt,a4paper]{article} 
\usepackage{amsmath,amssymb,amsfonts}
\usepackage[colorlinks]{hyperref}

\usepackage{relsize}
\usepackage{multirow}
\usepackage{diagbox}
\usepackage{booktabs}
\usepackage{thmtools,thm-restate}

\usepackage{graphics} % baraye resize e table

\usepackage{graphicx}

\usepackage{ cmll } %for \parr

\usepackage[noxy]{virginialake} %For table CFL_e for \vlinf{} and \vliinf{}

\usepackage{pifont}%
\usepackage[dvipsnames]{xcolor} %new colors

\usepackage{bussproofs}

\usepackage[bibliography=common]{apxproof}

\usepackage{tikz}
\usepackage{nicematrix}

\theoremstyle{plain} 
\newtheorem{theorem}{Theorem}

\newtheorem{lemma}[theorem]{Lemma}
\newtheorem{corollary}[theorem]{Corollary}

\theoremstyle{definition}
\newtheorem{definition}[theorem]{Definition}
\newtheorem{example}[theorem]{Example}

\newcommand{\st}{*}

\newcommand{\calL}{\mathcal{L}}

\newcommand{\FL}{\mathsf{FL_e}}
\newcommand{\F}{\mathbf{F}}

\newcommand{\BFL}{\mathbf{FL_e}}

\newcommand{\BIMALL}{\mathbf{IMALL}}

\newcommand{\BILL}{\mathbf{ILL}}

\newcommand{\BAIMALL}{\mathbf{IMALL_{w}}}

\newcommand{\BRIMALL}{\mathbf{IMALL_c}}

\newcommand{\BAILL}{\mathbf{ILLW}}

\newcommand{\BRILL}{\mathbf{ILLC}}

\newcommand{\BFLw}{\mathbf{FL_{ew}}}

\newcommand{\BFLc}{\mathbf{FL_{ec}}}
\newcommand{\LJ}{\mathbf{LJ}}

\newcommand{\LKb}{\mathbf{LK_b}}
\newcommand{\LKajab}{\mathbf{LK_!}}
\newcommand{\LJu}{\mathbf{LJ_u}}
\newcommand{\LJb}{\mathbf{LJ_b}}
\newcommand{\LJajab}{\mathbf{LJ_!}}

\newcommand{\bigast}{\mathop{\scalebox{1.5}{\raisebox{-0.2ex}{$\ast$}}}}%

\renewcommand{\phi}{\varphi}

\begin{document}

\title{Exponential Gaps Between Intuitionistic Linear Extended Frege Systems}
\author{Amirhossein Akbar Tabatabai\footnote{\texttt{amir.akbar@gmail.com}}\\
\small{Bernoulli Institute, University of Groningen$^*$}
}
\date{}
\maketitle
\begin{abstract}
In this paper, we establish exponential separations between Extended Frege systems for a range of intuitionistic substructural and linear logics. More precisely, for any logic $L$ below the intuitionistic logic obtained by extending $\mathbf{ILL}$ with structural rules, and any logic $M$ not contained in $L$, we construct a family of $\mathsf{FL_e}$-provable formulas that have short proofs in $M$-Frege but require proofs of exponential size in $L$-Extended Frege. The same result holds in the $!$-free settings, using $\mathbf{IMALL}$ and $\mathbf{FL_e}$ in place of $\mathbf{ILL}$. The key ingredient in proving these separations is a variant of the feasible disjunction property for $L$-Frege, which may be of independent interest.\\

\end{abstract}
\maketitle

\section{Introduction}

Proof complexity studies the size of proofs in various proof systems. For weak proof systems for classical logic, many interesting exponential lower bounds are known, including for the cut-free classical sequent calculus \cite{krajivcekfeasible2}, resolution \cite{Haken}, cutting planes \cite{Cutting2}, and bounded-depth Frege systems \cite{krajivcekFeasible}. However, for stronger proof systems, such as Frege systems or the sequent calculus $\mathbf{LK}$, proving super-polynomial lower bounds remains a major open problem of the field. For further details on known lower bounds, the methods used to establish them, and related open problems, see \cite{Krajicek}.

The situation is markedly different for non-classical logics, where exponential lower bounds are known for many strong proof systems. A landmark result in this direction is due to Hrubeš, who established exponential lower bounds on the number of lines in Frege systems, or equivalently on the size of Extended Frege proofs, for intuitionistic logic as well as several modal logics \cite{Hrubes1,Hrubes,Hrubes2}. Building on this foundation, Jeřábek \cite{jerabek} extended Hrubeš's method to a broader class of logics, including logics with infinite branching in both superintuitionistic and modal frameworks. For an accessible survey of techniques for proving such lower bounds, see \cite{AmirProofComp}.

These results were subsequently extended to substructural and linear logics. Jalali \cite{Raheleh} introduced a modification of the hard formulas used in Hrubeš's construction, yielding formulas that are already provable in the substructural logic $\mathsf{FL_e}$ while remaining, in a feasible sense, equivalent to the original formulas over intuitionistic Extended Frege. This allowed her to transfer the exponential lower bound to any proof system whose logic includes $\mathsf{FL_e}$ but is polynomially simulated by intuitionistic Extended Frege. In this way, the lower bound was extended to a broad family of intuitionistic substructural and linear logics. More recently, this lower bound was extended to the proof size of Frege systems for some classical substructural and linear logics, including the affine linear logic $\mathsf{LLW}$. This extension relies on a feasible form of conservativity of these classical systems over their intuitionistic counterparts \cite{tabatabai2026proof}.

Despite these advances, an important question remains open. Since the hardness in Jalali's construction ultimately comes from intuitionistic Extended Frege, her argument does not distinguish the proof complexity of $\mathsf{FL_e}$-Extended Frege from that of intuitionistic Extended Frege, or of intermediate systems, on formulas already provable in $\mathsf{FL_e}$.
The goal of this paper is to close this gap. We establish exponential separations between Extended Frege systems for a broad class of substructural and linear logics. 
To state our main result more precisely, let us call a logic \emph{structural} if it is obtained by extending $\mathbf{ILL}$ with structural rules, where a \emph{structural rule} is roughly a rule that only changes the number of occurrences of formulas, without introducing or eliminating any logical connective or constant \cite{ciabattoni2008axioms,ciabattoni2012algebraic}. We show that, for any structural logic $L$ below intuitionistic logic and any logic $M$ not contained in $L$, there is a family of $\mathsf{FL_e}$-provable formulas that have polynomial-size proofs in $M$-Frege but require exponentially many lines in $L$-Frege, or equivalently, exponential-size proofs in $L$-Extended Frege. The same result holds in the $!$-free setting, with $\mathbf{IMALL}$ and $\mathbf{FL_e}$ in place of $\mathbf{ILL}$.
In particular, we obtain exponential separations between $\mathsf{ILL}$-Extended Frege and each of $\mathsf{ILLW}$-Extended Frege and $\mathsf{ILLC}$-Extended Frege, as well as between $\mathsf{ILLW}$-Extended Frege and $\mathsf{IPC}$-Extended Frege, and between $\mathsf{ILLC}$-Extended Frege and $\mathsf{IPC}$-Extended Frege. These separations already hold in the $!$-free setting, with $\mathsf{IMALL}$ and $\mathsf{FL_e}$ in place of the $\mathsf{ILL}$ base.

The central technical tool underlying our separation result is a \emph{feasible disjunction property} for Frege systems of structural logics below intuitionistic logic. A Frege system satisfies the feasible disjunction property if, from any proof of a disjunction $A \vee B$, one can construct a proof of either $A$ or $B$ with only polynomial growth in the number of lines. This notion strengthens the classical disjunction property by requiring a form of feasibility. Feasible variants of the disjunction property have played an important role in proof complexity and have been used to establish several lower bounds \cite{Pudlak,PudlakBuss,BussMints}. Our proof of the feasible disjunction property builds on a refinement of the machinery developed in \cite{tabatabai2025}, which in turn is based on the machinery introduced by Hrubeš \cite{Hrubes}.

Finally, let us illustrate the role of the feasible disjunction property in separating proof systems. To this end, we sketch the proof in the case of $\mathsf{FL_e}$ and $\mathsf{IPC}$. The key idea is to turn a lower bound for a weaker logic into a separation result. Let $\{A_n\}_{n \in \mathbb{N}}$ be a family of formulas that requires exponentially many lines in $\mathsf{FL_e}$-Frege. Define
$
B_n := A_n \vee (\top \to 1).
$
Each $B_n$ is provable in $\mathsf{IPC}$ and has a short $\mathsf{IPC}$-Frege proof. Indeed, $\top \to 1$ has an $\mathsf{IPC}$-Frege proof of constant size, after which $B_n$ has a short proof by disjunction introduction. On the other hand, suppose that $B_n$ had an $\mathsf{FL_e}$-Frege proof with small number of lines. By the feasible disjunction property for $\mathsf{FL_e}$-Frege, this proof would yield, with only a polynomial increase in number of lines, either an $\mathsf{FL_e}$-Frege proof of $A_n$ or one of $\top \to 1$. Since $\top \to 1$ is not provable in $\mathsf{FL_e}$, the former must hold. This would give an $\mathsf{FL_e}$-Frege proof of $A_n$ with small number of lines, contradicting the assumed hardness of the family ${A_n}$. Thus, the feasible disjunction property turns the lower bound for $\{A_n\}$ into an exponential separation between $\mathsf{FL_e}$-Frege and $\mathsf{IPC}$-Frege in terms of the number of lines, or equivalently, between $\mathsf{FL_e}$-Extended Frege and $\mathsf{IPC}$-Extended Frege in terms of proof size.

\section{Preliminaries}
Let $\mathcal{L}_u=\{0,1,\wedge,\vee,\st,\to\}$ denote the language for unbounded substructural logics. We obtain its bounded extension by setting $\mathcal{L}_b=\mathcal{L}_u\cup\{\top,\bot\}$, and further extend this to the language of linear logic by adding \emph{exponentials}, namely $\mathcal{L}_!=\mathcal{L}_b\cup\{!\}$. Throughout, we write $\mathcal{L}\in\{\mathcal{L}_u,\mathcal{L}_b,\mathcal{L}_!\}$ when we refer uniformly to any of these languages. $\mathcal{L}_!$-formulas are defined inductively by the grammar:
\[
F \;::=\; p \mid 0 \mid 1 \mid \top \mid \bot \mid !F \mid F_1 \wedge F_2 \mid F_1 \vee F_2 \mid F_1 \to F_2 \mid F_1 \st F_2.
\]
The sets of $\mathcal{L}_u$-formulas and $\mathcal{L}_b$-formulas are defined analogously, by restricting the available constants and connectives accordingly.
By a standard abuse of notation, we identify a language $\mathcal{L}$ with the set of all $\mathcal{L}$-formulas.
We define the derived connective $\neg$ by
$\neg A := A \to 0$. 

\emph{Multisets of formulas} are defined as usual and we assume that they are always finite. We use capital Greek letters $\Gamma, \Delta, \dots$, as well as the bar notation as in $\overline{\phi}, \overline{\psi}, \overline{A}, \overline{B}, \dots$ to refer to multisets of formulas. For any two multisets $\Gamma$ and $\Pi$, by $(\Gamma, \Pi)$ or $\Gamma \cup \Pi$, we mean the union of $\Gamma$ and $\Pi$ as multisets. If $\Gamma= \{A_1,  \dots, A_n\}$, define $\bigast \Gamma:= A_1 * \dots * A_n$, and $\bigast \emptyset :=1$.

A \emph{logic} over a language $\mathcal L$ is any collection $L$ of $\mathcal L$-formulas satisfying the following closure conditions:
\begin{itemize}
\item \emph{Substitution:} $L$ is closed under substitutions;
\item \emph{Modus Ponens:} whenever $A,A\to B\in L$, then $B\in L$;
\item \emph{Adjunction:} whenever $A,B\in L$, then $A\wedge B\in L$;
\item \emph{Necessitation:} if $\mathcal L=\mathcal L_!$, then $A\in L$ implies $!A\in L$.
\end{itemize}

For further background on substructural and linear logics, see \cite{Avron,Ono,Troelstra}.
Following \cite{Cook}, a \emph{proof system} for a logic $L$ is a polynomial-time computable function $P$ whose range is precisely $L$. If $P(\pi)=A$, then $\pi$ is called a \emph{$P$-proof} of $A$, and we write
$P\vdash^\pi A$.
Let $P$ and $Q$ be proof systems for the logics $L_P$ and $L_Q$, respectively, where $L_P \subseteq L_Q$. We say that $Q$ \emph{size-simulates} $P$ (simulates, for short), and write $P \leq Q$, if every $P$-proof can be transformed into a $Q$-proof of the same formula with only a polynomial increase in size. More precisely, for every formula $A \in \mathcal{L}$ and every $P$-proof $\pi$ of $A$, there exists a $Q$-proof $\pi'$ of $A$ such that $|\pi'| \leq |\pi|^{O(1)}$.
The proof systems $P$ and $Q$ are said to be \emph{size-equivalent} (equivalent, for short) if they are proof systems for the same logic and each simulates the other.

\subsection{Sequent Calculi}

A \emph{single-conclusion sequent} (sequent, for short) over $\mathcal{L}$ is an expression of the form $\Gamma \Rightarrow \Delta$, where $\Gamma$ and $\Delta$ are finite \emph{multisets} of $\mathcal{L}$-formulas and $\Delta$ contains at most one $\mathcal{L}$-formula. The multiset $\Gamma$ is called the \emph{antecedent}, and $\Delta$ the \emph{succedent} of $S$. We also write $A \Leftrightarrow B$ for the pair of sequents $A \Rightarrow B$ and $B \Rightarrow A$. For a sequent $S=(\Gamma \Rightarrow \Delta)$, we define its \emph{interpretation} by $I(S)=(\bigast \Gamma \to \Delta)$ if $\Delta$ is a singleton, and $I(S)=(\bigast \Gamma \to 0)$, if $\Delta$ is empty.

\vspace{2pt}
\noindent A \emph{single-conclusion rule} (\emph{rule}, for short) over $\mathcal{L}$ is an expression 
\begin{center}
\AxiomC{$\mathcal{S}_1 \ldots \mathcal{S}_n$}
\UnaryInfC{$\mathcal{S}_0$}
\DisplayProof \qquad 
\end{center}
where $\mathcal{S}_i$'s are sequents over $\mathcal{L}$. We call $\mathcal{S}_1, \ldots, \mathcal{S}_n$ the \emph{premises} and $\mathcal{S}_0$ the \emph{conclusion} of the rule. An \emph{axiom} is a rule with no premises.

A \emph{sequent calculus} (or simply a \emph{calculus}) over $\mathcal{L}$ is a \emph{finite} set of rules over $\mathcal{L}$. Denote the language of a calculus $G$ by $\mathcal{L}_G$. For a sequent calculus $G$ and a finite set $X$ of rules, by $G+X$ or $GX$, we mean the  calculus consisting of all the rules in $G \cup X$.
We recall the standard sequent calculus $\mathbf{FL_e}$ in Table~\ref{Fig: sequent calculus}. We also consider the calculi introduced in Table~\ref{Fig: linear Seq cal}, obtained by adding the structural rules from Table~\ref{fig: structural rules} to $\mathbf{FL_e}$.

A \emph{(dag-like) $G$-proof} is defined as follows. Let $G$ be a sequent calculus and let $\mathcal{S}$ be a set of single-conclusion sequents. A $G$-proof $\pi$ of a sequent $S$ from $\mathcal{S}$ is a finite sequence
$\pi := S_1,\ldots,S_m$
of sequents such that $S_m = S$ and each $S_i$ is either an element of $\mathcal{S}$ or follows from earlier sequents in the sequence by an application of a rule of $G$. 
A $G$-proof is called \emph{tree-like} if each sequent is used as a premise of at most one rule in the proof.
We write $\mathcal{S} \vdash_G^\pi S$ when $\pi$ is a $G$-proof of $S$ from $\mathcal{S}$, and we write $\mathcal{S} \vdash_G S$ if such a proof exists. In the special case $\mathcal{S}=\varnothing$, we write $G \vdash S$. Two formulas $A$ and $B$ are said to be \emph{provably equivalent} if the sequent $A \Leftrightarrow B$ is derivable in $G$. The \emph{size} of a formula $A$ or a proof $\pi$, denoted $|A|$ and $|\pi|$, is the total number of symbols occurring in it. The \emph{number of lines} of a proof $\pi$, denoted $l(\pi)$, is the number of formulas appearing in $\pi$.

%The proof $\pi$ is called \emph{tree-like} if every sequent $S_j$ is used at most once as a premise of an inference rule.

For each sequent calculus $G \supseteq \mathbf{FL_e}$, we define the associated logic $L_G$ by
$L_G := \{A \in \mathcal{L} \mid G \vdash \, \Rightarrow A\}$. It is easy to see that $G \vdash S$ iff $I(S) \in L_G$, for any sequent $S$.
Throughout the paper, we denote calculi in boldface (e.g.\ $\mathbf{FL_e}$) and their associated logics in sans-serif (e.g.\ $L_{\mathbf{FL_e}}=\mathsf{FL_e}$). By abuse of notation, we use the same name for a calculus and its corresponding logic when they are considered over an extension of the underlying language. For instance, we denote the calculus $\mathbf{FL_e}$ over the language $\mathcal{L}_!$ by the same symbol $\mathbf{FL_e}$.

For the logic $L$ over $\calL_u$, we define its consequence relation $\vdash_{L}$ by stipulating that
\[
\Gamma \vdash_{L} A
\quad \text{iff} \quad
\{\, \Rightarrow \gamma \mid \gamma \in \Gamma \cup L\} \vdash_{\mathbf{FL_e}} \, \Rightarrow A,
\]
for every $\Gamma \cup \{A\} \subseteq \mathcal{L}$. Similarly, we define $\vdash_L$ for logics over $\mathcal{L}_b$ (resp. $\mathcal{L}_!$) using $\BIMALL$ (resp. $\BILL$) in place of $\mathbf{FL_e}$.

\begin{table}[t]
\centering
\renewcommand{\arraystretch}{1.2}

\begin{tabular}{ccc}
\AxiomC{$ $}
\RightLabel{\scriptsize(id)}
\UnaryInfC{$A \Rightarrow A$}
\DisplayProof \qquad
&
\AxiomC{$ $}
\RightLabel{\scriptsize($1w$)}
\UnaryInfC{$ \Rightarrow 1 $}
\DisplayProof \qquad
&
\AxiomC{$ $}
\RightLabel{\scriptsize($0w$)}
\UnaryInfC{$0 \Rightarrow $}
\DisplayProof
\\[2.5ex]
\end{tabular}

\begin{tabular}{cc}
\AxiomC{$\Gamma \Rightarrow \Delta$}
\RightLabel{\scriptsize $(1w)$}
\UnaryInfC{$\Gamma, 1 \Rightarrow \Delta$}
\DisplayProof \qquad
&
\AxiomC{$\Gamma \Rightarrow $}
\RightLabel{\scriptsize($0w$)}
\UnaryInfC{$\Gamma \Rightarrow 0 $}
\DisplayProof \qquad
\\[2.5ex]
\end{tabular}

\begin{tabular}{cc}

\AxiomC{$\Gamma, A_i \Rightarrow \Delta$}
\RightLabel{\scriptsize($L\wedge_i$)}
\UnaryInfC{$\Gamma, A_1\wedge A_2 \Rightarrow \Delta$}
\DisplayProof
&
\AxiomC{$\Gamma \Rightarrow A$}
\AxiomC{$\Gamma \Rightarrow B$}
\RightLabel{\scriptsize($R\wedge$)}
\BinaryInfC{$\Gamma \Rightarrow A\wedge B$}
\DisplayProof
\\[3ex]

\AxiomC{$\Gamma,A \Rightarrow \Delta$}
\AxiomC{$\Gamma,B \Rightarrow \Delta$}
\RightLabel{\scriptsize($L\vee$)}
\BinaryInfC{$\Gamma,A\vee B \Rightarrow \Delta$}
\DisplayProof
&
\AxiomC{$\Gamma \Rightarrow A_i$}
\RightLabel{\scriptsize($R\vee_i$)}
\UnaryInfC{$\Gamma \Rightarrow A_0\vee A_1$}
\DisplayProof
\\[3ex]

\AxiomC{$\Gamma,A,B \Rightarrow \Delta$}
\RightLabel{\scriptsize($L*$)}
\UnaryInfC{$\Gamma,A*B \Rightarrow \Delta$}
\DisplayProof
&
\AxiomC{$\Gamma \Rightarrow A$}
\AxiomC{$\Sigma \Rightarrow B$}
\RightLabel{\scriptsize($R*$)}
\BinaryInfC{$\Gamma,\Sigma \Rightarrow A*B$}
\DisplayProof
\\[3ex]

\AxiomC{$\Gamma \Rightarrow A$}
\AxiomC{$\Sigma,B \Rightarrow \Delta$}
\RightLabel{\scriptsize($L\to$)}
\BinaryInfC{$\Gamma,\Sigma,A\to B \Rightarrow \Delta$}
\DisplayProof
&
\AxiomC{$\Gamma,A \Rightarrow B$}
\RightLabel{\scriptsize($R\to$)}
\UnaryInfC{$\Gamma \Rightarrow A\to B$}
\DisplayProof
\\[3ex]

\end{tabular}

\begin{tabular}{c}

\AxiomC{$\Gamma \Rightarrow A$}
\AxiomC{$\Sigma,A \Rightarrow \Delta$}
\RightLabel{\scriptsize(cut)}
\BinaryInfC{$\Gamma,\Sigma \Rightarrow \Delta$}
\DisplayProof
\end{tabular}

\vspace{5pt}
\caption{Sequent calculus $\mathbf{FL_e}$.}
\label{Fig: sequent calculus}
\end{table}

\begin{table}[ht]
\centering
\small
\renewcommand{\arraystretch}{1.8}
\setlength{\tabcolsep}{10pt}

\begin{tabular}{|l|c|c|}
\hline
\textbf{Axioms} $(\bot),(\top)$ &
$\Gamma \Rightarrow \top$ \ $(\top)$ &
$\Gamma,\bot \Rightarrow \Delta$ \ $(\bot)$ \\
\hline

\textbf{Weakening} $(w)$ &
\scalebox{0.95}{
  \AxiomC{$\Gamma \Rightarrow \Delta$}
  \RightLabel{$(Lw)$}
  \UnaryInfC{$\Gamma,A \Rightarrow \Delta$}
  \DisplayProof
} &
\scalebox{0.95}{
  \AxiomC{$\Gamma \Rightarrow$}
  \RightLabel{$(Rw)$}
  \UnaryInfC{$\Gamma \Rightarrow A$}
  \DisplayProof
} \\
\hline

\textbf{Contraction} $(c)$ &
\multicolumn{2}{c|}{
\scalebox{0.95}{
  \AxiomC{$\Gamma,A,A \Rightarrow \Delta$}
  \RightLabel{$(Lc)$}
  \UnaryInfC{$\Gamma,A \Rightarrow \Delta$}
  \DisplayProof
}} \\
\hline

\multirow{2}{*}{\textbf{Exponential rules} $(exp)$} &
\scalebox{0.95}{
  \AxiomC{$!\Gamma \Rightarrow A$}
  \RightLabel{$(R!)$}
  \UnaryInfC{$!\Gamma \Rightarrow !A$}
  \DisplayProof
} &
\scalebox{0.95}{
  \AxiomC{$\Gamma,A \Rightarrow \Delta$}
  \RightLabel{$(L!)$}
  \UnaryInfC{$\Gamma,!A \Rightarrow \Delta$}
  \DisplayProof
} \\
\cline{2-3}
&
\scalebox{0.95}{
  \AxiomC{$\Gamma \Rightarrow \Delta$}
  \RightLabel{$(W!)$}
  \UnaryInfC{$\Gamma,!A \Rightarrow \Delta$}
  \DisplayProof
} &
\scalebox{0.95}{
  \AxiomC{$\Gamma,!A,!A \Rightarrow \Delta$}
  \RightLabel{$(C!)$}
  \UnaryInfC{$\Gamma,!A \Rightarrow \Delta$}
  \DisplayProof
} \\
\hline
\end{tabular}

\caption{Other rules}
\label{fig: structural rules}
\end{table}

\begin{table}[ht]
\centering
\small
\renewcommand{\arraystretch}{1.5}
\setlength{\tabcolsep}{8pt}

\begin{tabular}{|c|c|c|c|c|}
\hline
 & $G$ & $G + \{w\}$ & $G + \{c\}$ & $G + \{w,c\}$ \\ \hline
$G$ & $\BFL$ & $\BFLw$ & $\BFLc$ & $\LJu$ \\ \hline
$G+ \{(\bot),(\top)\}$ & $\BIMALL$ & $\BAIMALL$ & $\BRIMALL$ & $\LJb$ \\ \hline
$G+ \{exp\}$ & $\BILL$ & $\BAILL$ & $\BRILL$ & $\LJajab$ \\ \hline
\end{tabular}

\caption{Overview of sequent calculi.}
\label{Fig: linear Seq cal}
\end{table}

%We recall the sequent calculus $\LK$ for classical propositional logic $\CPC$ over the language $\mathcal{L}_p=\{\top,\bot,\wedge,\vee,\to\}$:
%\[
%\LK := \{(id),(\top),(\bot),(Lw),(Rw),(Lc),(Rc),(L\circ),(R\circ),(cut)\},
%\]
%where $\circ \in \{\wedge,\vee,\to\}$ and the rules are those given in Tables~\ref{Fig: sequent calculus} and~\ref{fig: structural rules}. The restriction of $\LK$ to single-conclusion sequents is denoted by $\LJ$ and corresponds to the standard sequent calculus for intuitionistic propositional logic $\IPC$.

%It is a routine but essential observation that proofs in any two of the calculi $\LK$, $\LKu$, $\LKb$, and $\LKajab$ over a common fragment of the language can be translated from one system to another with at most a polynomial increase in size and number of inference steps; the same holds for $\LJ$. This relies on the fact that, in the presence of weakening and contraction, the connective $*$ behaves equivalently to $\wedge$, while $0$ and $1$ correspond to $\bot$ and $\top$, respectively. Moreover, the modalities $!$ and $?$ can be interpreted as the identity operation, i.e.\ $!A := ?A := A$, which ensures that all rules of the extended systems remain admissible under this interpretation.

\subsection{Frege Systems}

We now introduce Frege systems for substructural and linear logics. An \emph{inference system} $F$ over $\mathcal{L}$ consists of rules of the form $A_1,\ldots,A_n / A$, where all $A_i$ and $A$ are formulas in $\mathcal{L}$.
A \emph{dag-like $F$-proof} (or simply an $F$-proof) of a formula $A$ from a set of formulas $\mathcal{A}$ is a finite sequence $\pi := A_1,\ldots,A_m$ of $\mathcal{L}$-formulas such that $A_m=A$ and each $A_i$ is either an element of $\mathcal{A}$ or obtained from earlier formulas by an instance of a rule of $F$. 
%The proof $\pi$ is called \emph{tree-like} if every line is used at most once as a premise of an inference rule. 
Each $A_i$ is called a \emph{line} of $\pi$, and we write $l(\pi)$ for the number of lines; clearly $l(\pi)\le |\pi|$.
We write $\mathcal{A} \vdash_F^{\pi} A$ if $\pi$ is an $F$-proof of $A$ from $\mathcal{A}$, and $\mathcal{A} \vdash_F A$ if such a proof exists. When $\mathcal{A}=\emptyset$, we simply write $F \vdash A$, and the set of all provable formulas is called the \emph{logic induced by $F$}.

Let $L \subseteq M$ be two logics and $P$ and $Q$ be either an inference system or a sequent system for $L$ and $M$, respectively. We say that $Q$ \emph{line-simulates} $P$, and write $P \leq_l Q$, if every $P$-proof can be transformed into a $Q$-proof of the same formula with only a polynomial increase in number of lines. More precisely, for every formula $A \in \mathcal{L}$ and every $P$-proof $\pi$ of $A$, there exists a $Q$-proof $\pi'$ of $A$ such that $l(\pi') \leq l(\pi)^{O(1)}$.
The proof systems $P$ and $Q$ are said to be \emph{line-equivalent} if they are proof systems for the same logic and each line-simulates the other.  

Let $L$ be a substructural or linear logic over $\mathcal{L}$. A finite inference system $F$ is called a \emph{Frege system} for $L$ (an $L$-$\F$ system) if it is sound, i.e.\ $F \vdash A$ implies $A \in L$, and strongly complete, i.e.\ $A_1,\ldots,A_n \vdash_L A$ implies $A_1,\ldots,A_n \vdash_F A$. A Frege system is called \emph{standard} if also $A_1,\ldots,A_n \vdash_F A$ implies $A_1,\ldots,A_n \vdash_L A$. Following \cite{jerabek}, we assume all Frege systems are standard.

As examples, Table \ref{Fig: FL_e} presents a Frege system for $\mathsf{FL_e}$. By adding the corresponding axioms from Table \ref{Fig: linear}, we obtain Frege systems for the logics whose sequent calculi are presented in Table \ref{Fig: linear Seq cal}.
\begin{table}[h]
\centering
\normalsize
\renewcommand{\arraystretch}{0.95}
\begin{tabular}{|c|c|}
\hline
(id) & $A \rightarrow A$ \\
\hline
(pf) & $(A \rightarrow B) \rightarrow \big((C \rightarrow A) \rightarrow(C \rightarrow B) \big)$  \\
\hline
(per) & $\big(A \rightarrow(B \rightarrow C)\big) \rightarrow \big(B \rightarrow(A \rightarrow C) \big)$  \\
\hline
$(* \wedge)$ & $\big((A \wedge 1) * (B \wedge 1)\big) \rightarrow(A \wedge B)$ \\
\hline
$(\wedge \! \to)_1$ & $(A \wedge B) \rightarrow A$ \\
\hline
$(\wedge \! \to)_2$ & $(A \wedge B) \rightarrow B$  \\
\hline
$(\rightarrow \! \wedge)$ & $\big((A \rightarrow B) \wedge(A \rightarrow C)\big) \rightarrow \big(A \rightarrow(B \wedge C) \big)$ \\
\hline
$(\rightarrow \! \vee)_1$ & $A \rightarrow(A \vee B)$  \\
\hline
$(\rightarrow \! \vee)_2$ & $B \rightarrow(A \vee B)$  \\
\hline
$(\vee \! \rightarrow)$ & $\big((A \rightarrow C) \wedge(B \rightarrow C)\big) \rightarrow (A \vee B \rightarrow C)$ \\
\hline
$(\rightarrow \! *)$ & $B \rightarrow(A \rightarrow (A \st B))$  \\
\hline
$(* \! \rightarrow)$ & $\big(B \rightarrow(A \rightarrow C)\big) \rightarrow((A \st B) \rightarrow C)$  \\
\hline
(1) & $1$  \\
\hline
$(1 \! \rightarrow)$ & $1 \rightarrow(A \rightarrow A)$  \\
\hline
$(\mathrm{mp})$ & $\vliinf{}{}{B}{A}{A \to B}$ \\
\hline
$(\operatorname{adj}_u)$ & $\vlinf{}{}{A \wedge 1}{A}$ \\
\hline
\end{tabular}
\vspace{3pt}
\caption{The system $\FL$-$\F$.}
\label{Fig: FL_e}
\end{table}

\begin{table}[tbh]
\centering
\normalsize
\renewcommand{\arraystretch}{0.95}
\begin{tabular}{|c|c|}
\hline
$(\text{top})$ & $A \to \top$\\
\hline
$(\text{bot})$ & $\bot \to A$ \\
\hline
$(\text{dn})$ & $\neg \neg A \to A$ \\
\hline
$(w)$ & $A \to (B \to A)$ \\
\hline
$(c)$ & $(A \to (A \to B)) \to (A \to B)$ \\
\hline
$(!w)$ & $A \to (!B \to A)$\\
\hline
$(!c)$ & $(!A \to (!A \to B)) \to (!A \to B)$\\
\hline
$(!\text{K})$ & $!(A \to B) \to (!A \to !B)$\\
\hline
$(!\text{T})$ & $!A \to A$\\
\hline
$(!\text{4})$ & $!A \to !!A$\\
\hline
$(\text{nec})$ & $\vlinf{}{}{!A}{A}$\\
\hline
\end{tabular}
\vspace{3pt}
\caption{Additional axioms and rules.}
\label{Fig: linear}
\end{table}

\begin{theorem}\cite[Lemma 3.11]{Raheleh}\label{Lem: Equivalence Of Frege}
Let $L \subseteq M$ be logics, let $P$ be a Frege system for $L$, and let $Q$ be a Frege system for $M$. Then, $P \leq_l Q$ and $P \leq Q$. Consequently, all Frege systems for any logic are both size-equivalent and line-equivalent.
\end{theorem}

Using Theorem \ref{Lem: Equivalence Of Frege}, any two Frege systems for a logic $L$ are equivalent, and we refer to \emph{the} Frege system for $L$, denoted $L$-$\F$. There is a similar and stronger notion of an Extended Frege system for a logic $L$, whose proof size is polynomially related to the number of lines in the corresponding Frege system. For simplicity, we do not define these systems here and instead prove our separation results in terms of the number of lines in Frege systems.

Finally, let $P$ be either a sequent calculus or a Frege system. We say that $P$ has the \emph{feasible disjunction property} if, for any $P$-proof $\pi$ of $A \vee B$, there is a $P$-proof $\sigma$ of either $A$ or $B$ such that
$
l(\sigma) \leq l(\pi)^{O(1)}.
$

\section{Structural Calculi}

In this section, we define structural rules and structural calculi. We essentially follow the framework of \cite{ciabattoni2008axioms,ciabattoni2012algebraic}. However, we use a more restrictive notion of structural rules, allowing free context wherever possible. This restriction is motivated by our desire for each structural rule to be equivalent to an axiom.

\begin{definition}
Let $\Gamma_i$'s, $\Delta_i$'s, and $\Pi_j$'s be pairwise disjoint families of pairwise distinct multiset variables, for $1 \leq i \leq n$ and $1 \leq j \leq m$. Assume  $\bar{\mu}_{ir}$'s, $\bar{\nu}_{js}$'s, $\bar{\rho}_{js}$'s, and, $\bar{\mu}$ are multisets of \emph{atomic} formulas and $\nu$ is an \emph{atom}, where $1 \leq r \leq k_i$ and $1 \leq s \leq l_j$. If we choose the value $0$ for $n$, $m$, $k_i$, or $l_j$, we mean $i$, $j$, $r$, or $s$ range over the empty set, respectively, and if there is no fear of confusion, we omit the domain of these indices. A \emph{structural rule} is a rule that has one of the following forms:\\

\noindent $\bullet$
\emph{left structural}:
\small\begin{center}
 \AxiomC{$\{\Pi_j , \bar{\nu}_{js} \Rightarrow \bar{\rho}_{js} \mid 1 \leq j \leq m, 1 \leq s \leq l_j \}$} 
 \AxiomC{$\{\Gamma_i , \bar{\mu}_{ir} \Rightarrow \Delta_i \mid 1 \leq i \leq n, 1 \leq r \leq k_i \}$}
 \BinaryInfC{$\Pi_1, \dots, \Pi_m, \Gamma_1, \dots, \Gamma_n, \bar{\mu} \Rightarrow \Delta_1, \dots, \Delta_n $}
 \DisplayProof
\end{center}
\normalsize where 
$\bar{\rho}_{js}$ has at most one formula, for each $j$ and $s$. As the rule is single-conclusion, at most one of $\Delta_i$'s in each instance of the rule can be substituted by a formula and the rest are empty. If $n=0$ (resp. $m=0$), there is no premise of the right (resp. left) branch and the conclusion is of the form $(\Pi_1, \dots, \Pi_m, \bar{\mu} \Rightarrow \,)$ (resp. $\Gamma_1, \dots, \Gamma_n, \bar{\mu} \Rightarrow \Delta_1, \dots, \Delta_n$). Note that if $m=n=0$, the rule has no premise and its conclusion is $(\bar{\mu} \Rightarrow \,)$. \\

\noindent $\bullet$
\emph{right structural}:
\small\begin{center}
 \AxiomC{$\{\Gamma_i , \bar{\mu}_{ir} \Rightarrow \bar{\nu}_{ir} \mid 1 \leq i \leq n, 1 \leq r \leq k_i\}$}
 \UnaryInfC{$\Gamma_1, \dots, \Gamma_n, \bar{\mu} \Rightarrow \nu $}
 \DisplayProof
\end{center}
\normalsize where $\bar{\nu}_{ir}$ has at most one formula, for each $i$ and $r$. Note that if $n=0$, the rule has no premise and its conclusion is $(\bar{\mu} \Rightarrow \nu)$. 
    
\end{definition}

\begin{example}\label{exam StructuralRules}
The following rules are examples of structural rules:

\begin{center}
\begin{tabular}{c c c c}

% (i) and (o)
\AxiomC{$\Gamma \Rightarrow \Delta$}
\RightLabel{{\footnotesize $(Lw)$}}
\UnaryInfC{$\Gamma, p \Rightarrow \Delta$}
\DisplayProof
&
\AxiomC{$\Gamma \Rightarrow $}
\RightLabel{{\footnotesize $(Rw)$}}
\UnaryInfC{$\Gamma \Rightarrow p$}
\DisplayProof
&
% (c) and (exp)
\AxiomC{$\Gamma, p, p \Rightarrow \Delta$}
\RightLabel{{\footnotesize $(Lc)$}}
\UnaryInfC{$\Gamma, p \Rightarrow \Delta$}
\DisplayProof
&
\AxiomC{$\Gamma, p \Rightarrow \Delta$}
\RightLabel{{\footnotesize $(\mathrm{exp})$}}
\UnaryInfC{$\Gamma, p, p \Rightarrow \Delta$}
\DisplayProof
\\[1em]
\end{tabular}

\begin{tabular}{c c}
% (anl-knot^n_m) alone
\AxiomC{$\Gamma, \overbrace{p, \ldots, p}^{m} \Rightarrow \Delta$}
\RightLabel{{\footnotesize $(\mathrm{knot}^n_m)$}}
\UnaryInfC{$\Gamma, \underbrace{p, \ldots, p}_{n} \Rightarrow \Delta$}
\DisplayProof
&
\AxiomC{$\{\Gamma, p_{i_1}, \ldots, p_{i_m} \Rightarrow \Delta\}_{i_1,\ldots,i_m \in \{1,\ldots,n\}}$}
\RightLabel{{\footnotesize $(\mathrm{anl}\text{-}\mathrm{knot}^n_m)$}}
\UnaryInfC{$\Gamma, p_1, \ldots, p_n \Rightarrow \Delta$}
\DisplayProof
\\
\end{tabular}

\begin{tabular}{c c}
% (min) and (mix)
\AxiomC{$\Gamma, p \Rightarrow \Delta$}
\AxiomC{$\Gamma, q \Rightarrow \Delta$}
\RightLabel{{\footnotesize $(\mathrm{min})$}}
\BinaryInfC{$\Gamma, p, q \Rightarrow \Delta$}
\DisplayProof
&
\AxiomC{$\Pi \Rightarrow$}
\AxiomC{$\Gamma, p \Rightarrow \Delta$}
\RightLabel{{\footnotesize $(\mathrm{mix})$}}
\BinaryInfC{$\Pi, \Gamma, p \Rightarrow \Delta$}
\DisplayProof\\[1em]
\end{tabular}

\begin{tabular}{ c}
\AxiomC{$\Pi \Rightarrow p$}
\AxiomC{$\Pi, p \Rightarrow $}
\RightLabel{{\footnotesize $(wc)$}}
\BinaryInfC{$\Pi \Rightarrow $}
\DisplayProof
\end{tabular}
\end{center}
The first three in the first line are the usual structural rules in the single-conclusion setting. 
The rules $(\mathrm{knot}^n_m)$ and $(\mathrm{anl}\text{-}\mathrm{knot}^n_m)$ are different rule versions of the axiom $p^n \Rightarrow p^m$ that generalizes both $(Lw)$ ($m=0$ and $n=1$) and $(Lc)$  ($m=2$ and $n=1$).
\end{example}

\begin{definition}\label{def: structural calculus}
A calculus $G$ over $\mathcal{L}$ is called  \emph{structural} if it extends the calculus
\begin{itemize}
    \item 
$\mathbf{FL_e}$, if $\mathcal{L}=\mathcal{L}_u$,
    \item 
$\mathbf{IMALL}$, if $\mathcal{L}=\mathcal{L}_b$,
    \item 
$\mathbf{ILL}$, if $\mathcal{L}=\mathcal{L}_!$,
\end{itemize}
by a finite set of structural rules. A logic is called structural iff it has a structural calculus.
\end{definition}

\begin{example}
The systems $\mathbf{FL_e}$, $\mathbf{FL_{ew}}$, $\mathbf{FL_{ec}}$, and $\LJb$ are all structural calculi over $\mathcal{L}_u$. Similarly, replacing the base system $\mathbf{FL_e}$ with $\mathbf{IMALL}$ and $\mathbf{ILL}$, respectively, yields structural calculi over the languages $\mathcal{L}_b$ and $\mathcal{L}_!$.
\end{example}

\begin{lemma}\label{lem: RuleToAxiom}
Every structural rule is equivalent to an axiom over $\mathbf{FL_e}$.
\end{lemma}
\begin{proof}
The proof is an easy and well-known procedure; see \cite[Theorem 4.2]{ciabattoni2008axioms}.
\end{proof}

\begin{corollary}\label{Cor: ExistenceOfFrege}
\begin{itemize}
    \item[$(i)$] 
Every structural logic has a Frege system.
    \item[$(ii)$] 
For any structural calculus $G$, if $L_G \subseteq \mathsf{LK_!}$, then $\mathbf{LK}_!$ proves all instances of all structural rules in $G$.   
\end{itemize}
\end{corollary}
\begin{proof}
For $(i)$, using Lemma \ref{lem: RuleToAxiom}, we can transform each of the added structural rules into axioms and add them to the base system. For $(ii)$, for any structural rule, it suffices to prove the corresponding axiom $A$ in $\LKajab$. Since $G$ extends $\mathbf{FL_e}$, we have $G \vdash \, \Rightarrow A$. Hence, $A \in L_G \subseteq \mathsf{LK_!}$ which implies $\LKajab \vdash \, \Rightarrow A$.
\end{proof}

\begin{theorem}\label{Lem: Equivalence Of Frege}
For any structural logic $L$, all structural calculi for $L$, their tree-like versions, and Frege systems for $L$ are size-equivalent and line-equivalent.
\end{theorem}
\begin{proof}
The claim is easy and analogous to \cite[Theorem 6.3]{Raheleh}. For further details, see \cite{Krajicek}.
\end{proof}

\section{Feasible Disjunction Property}

In this section, we prove that any structural calculus $G$ such that $L_G \subseteq \mathsf{LK_!}$ enjoys the feasible disjunction property. That is, for any $G$-proof $\pi$ of $(\, \Rightarrow A \vee B)$, there is a $G$-proof $\sigma$ of either $(\, \Rightarrow A)$ or $(\, \Rightarrow B)$ such that
$
l(\sigma) \leq l(\pi)^{O(1)}.
$
Using the line-equivalence between $G$ and $L_G$-Frege, this also yields the feasible disjunction property for $L_G$-Frege.

\subsection{The Translation}

Fix $G$ to be a structural calculus such that $L_G \subseteq \mathsf{LK_!}$. To prove the feasible disjunction property for $G$, we first introduce a translation function that transforms $G$-proofs of disjunctions into $\LKajab$-proofs of disjunctions with atomic disjuncts. The latter form is more amenable to establishing the disjunction property.

\begin{definition}  \label{Translation}
Given any formula $A \in \mathcal{L}$, let $\langle A \rangle$ be a fresh propositional variable corresponding to $A$. We call $\langle A \rangle$ an \emph{angled atom}. The language obtained by adjoining all angled atoms to $\mathcal{L}$ is denoted by $\mathcal{L}^{+}$. The translation function $t:\mathcal{L} \to \mathcal{L}^{+}$ is defined as follows:
\begin{itemize}
\item[$\bullet$] 
$c^t=c \wedge \langle c \rangle$, for any constant $c \in \mathcal{L}$;
\item[$\bullet$] 
$p^t= \langle p \rangle$, for any atomic formula $p \in \mathcal{L}$;
\item[$\bullet$]
$(A \circ B) ^t = (A^t \circ B^t) \wedge \langle A \circ B \rangle$, for any $\circ \in \{*, \wedge, \vee, \to\}$;
\item[$\bullet$]
$(!A)^t =(!A^t) \wedge \langle !A \rangle$, if $\mathcal{L}=\mathcal{L}_!$.
\end{itemize}
For a multiset $\Gamma$, define  $\Gamma^t = \{\gamma^t \mid \gamma \in \Gamma\}$. 
The \emph{standard substitution} $s: {\mathcal{L}}^+ \to \mathcal{L}$ is defined as:
\begin{itemize}
\item[$\bullet$]
$\langle A \rangle^s = A$, and $p^s=p$, for any formula $A \in \mathcal{L}$ and atom $p \in \mathcal{L}$;
\item[$\bullet$]
$(A \circ B) ^s = A^s \circ B^s$, for any $\circ \in \{*, \wedge, \vee, \to\}$;
\item[$\bullet$]
$(! A)^s =! A^s$, if $\mathcal{L}=\mathcal{L}_!$.
\end{itemize}
For a multiset $\Gamma$, define $\Gamma^s =\{\gamma^s \mid \gamma \in \Gamma\}$.  
\end{definition}
The map $s:\mathcal{L}^{+}\to\mathcal{L}$ eliminates angled atoms by replacing each occurrence of $\langle A \rangle$ with the corresponding formula $A$, while leaving all non-angled atoms unchanged. By a straightforward induction on the structure of formulas, one obtains that for every formula $A(p_1,\dots,p_n)\in\mathcal{L}$,
$
A(\langle B_1\rangle,\dots,\langle B_n\rangle)^s
= A(B_1,\dots,B_n)$.
In this sense, $s$ serves as a reverse operation from $\mathcal{L}^{+}$ back to $\mathcal{L}$, undoing the effect of the translation $t$ and recovering the original syntactic structure.

\begin{lemma}\label{TranslationAndAtoms}
$\LKajab \vdash A^t \Rightarrow \langle A \rangle$, for every formula $A \in \mathcal{L}_!$. 
\end{lemma}
\begin{proof}
This follows directly from Definition~\ref{Translation}.
\end{proof}

\begin{definition}\label{ImplicationalHorn}
A sequent is called \emph{Horn} if it is of the form $\Gamma \Rightarrow p$, where $p$ is an atomic formula and every formula in $\Gamma$ is a non-empty conjunction of occurrences of $1$ and atomic formulas.
\end{definition}

\begin{example}
Let $p, q, r, s$ be atomic formulas. Then the sequent $(p \wedge q, r \wedge 1, 1 \wedge 1 \Rightarrow s)$ is Horn.
\end{example}

The following theorem shows that the translation $t$ preserves provability in $G$ if we allow some ``harmless'' sequents among the assumptions. By harmless, we mean either a sequent of the form $\langle !A \rangle \Rightarrow !\langle !A \rangle$ or a Horn sequent whose standard substitution is $G$-provable. For feasibility reasons, we need a feasible version of this preservation theorem:

\begin{theorem}(Provability Preservation) \label{MainTheorem}
Let $G$ be a structural calculus and $L_G \subseteq \mathsf{LK_!}$. Then,
given a $G$-proof $\pi$ of $\Gamma \Rightarrow \Delta$, there is a set $\mathcal{S}_{\pi}$ of sequents and a $G$-proof $\sigma_{\pi}$ such that:
\begin{itemize}
\item[$(i)$]
$\mathcal{S}_{\pi} \vdash_{\LKajab} \Gamma^t \Rightarrow  \Delta^t$,
\item[$(ii)$]
sequents in $\mathcal{S}_{\pi}$ are either Horn or in the form $\langle ! A \rangle \Rightarrow ! \langle ! A \rangle$ when $\mathcal{L}=\mathcal{L}_!$,
\item[$(iii)$]
$G\vdash^{\sigma_{\pi}} \; \Rightarrow \bigwedge I(\mathcal{S}_{\pi}^{s})$, where $s$ is the standard substitution.
\item[$(iv)$]
$|\mathcal{S}_{\pi}|, l(\sigma_{\pi}) \leq l(\pi)^{O(1)}$.
\end{itemize}
\end{theorem}
\begin{proof}
We only present the case $\mathcal{L}=\mathcal{L}_!$, as the cases for the fragments are subsumed. W.l.o.g., we may also assume that $\pi$ is tree-like by Theorem \ref{Lem: Equivalence Of Frege}. We construct the set $\mathcal{S}_{\pi}$ and the $G$-proof $\sigma_{\pi}$ by recursion on the structure of $\pi$. We address the size issues later:

\vspace{4pt}
\noindent $\bullet$ If $\pi$ is an axiom $(A \Rightarrow A)$, the axiom $(0 \Rightarrow \,)$, or $(\Sigma, \bot \Rightarrow \Lambda)$, then it is clear that $\LKajab \vdash \Gamma^t \Rightarrow \Delta^t$. Therefore, we can simply set $\mathcal{S}_{\pi}$ to be the empty set and $\sigma_{\pi}$ to be the axiom $(\, \Rightarrow \top)$. If $\pi$ is the axiom $(\, \Rightarrow 1)$, then, as $1^t = 1 \wedge \langle 1 \rangle$, it suffices to set $\mathcal{S}_{\pi} = \{\, \Rightarrow \langle 1 \rangle\}$. Then $(\, \Rightarrow 1^t)$ is easily provable in $\LKajab$. Condition $(ii)$ is immediate, since $(\, \Rightarrow \langle 1 \rangle)$ is Horn. For $\sigma_{\pi}$, we can use the canonical $G$-proof of $(\, \Rightarrow 1 \to 1)$.
If $\pi$ is the axiom $(\Gamma \Rightarrow \top)$, set $\mathcal{S}_{\pi} = \{\, \Rightarrow \langle \top \rangle\}$. As $\LKajab \vdash \Gamma^t \Rightarrow \top$, it suffices to prove $\Gamma^t \Rightarrow \langle \top \rangle$ from $\mathcal{S}_{\pi}$, which is possible in $\LKajab$. Condition $(ii)$ is immediate, since $(\, \Rightarrow \langle \top \rangle)$ is Horn. For $\sigma_{\pi}$, it suffices to take the canonical $G$-proof of $(\, \Rightarrow 1 \to \top)$.

\vspace{4pt}
\noindent $\bullet$ The cases for structural rules are straightforward, as the translation $t$ commutes with these rules, and all such rules are provable in $\LKajab$ by the assumption $L_G \subseteq \mathsf{CPC}$ and Corollary \ref{Cor: ExistenceOfFrege}, part $(ii)$.

\vspace{4pt}
\noindent $\bullet$ The cases of the rules $(W!)$ and $(C!)$ can be easily handled using the stronger rules $(Lw)$ and $(Lc)$ in $\LKajab$, respectively.

\vspace{4pt}
\noindent $\bullet$ The cases of the left logical rules are straightforward. We only present the case $(L*)$, as the others are similar. Note that we consider $(1w)$ and $(L!)$ as left rules. For the case $(L*)$, the proof $\pi$ must be of the form:
\begin{center}
\begin{tabular}{c}
\AxiomC{$\pi_1$}
\noLine
\UnaryInfC{$\Gamma, A, B \Rightarrow \Delta$}
\UnaryInfC{$\Gamma, A * B \Rightarrow \Delta$}
\DisplayProof
\end{tabular}
\end{center}
Define $\mathcal{S}_\pi=\mathcal{S}_{\pi_1}$. By the induction hypothesis, we have
$\mathcal{S}_{\pi_1} \vdash_{\LKajab} \Gamma^t, A^t, B^t \Rightarrow \Delta^t$.
For $(i)$, by the rule $(L*)$, we obtain:
\begin{center}
\begin{tabular}{c}
\AxiomC{$\Gamma^t, A^t, B^t \Rightarrow \Delta^t$}
\UnaryInfC{$\Gamma^t, A^t * B^t \Rightarrow \Delta^t$}
\DisplayProof
\end{tabular}
\end{center}
As $\LKajab \vdash (A*B)^t \Rightarrow A^t * B^t$, by cut we have
$\mathcal{S}_{\pi} \vdash_{\LKajab} \Gamma^t, (A*B)^t \Rightarrow \Delta^t$.
Property $(ii)$ follows immediately from the induction hypothesis. For $(iii)$, it suffices to set
$\sigma_{\pi}=\sigma_{\pi_1}$.

\vspace{4pt}
\noindent $\bullet$ If the last rule in $\pi$ is $(0w)$, then $\pi$ is of the form:
\begin{center}
\begin{tabular}{c}
\AxiomC{$\pi_1$}
\noLine
\UnaryInfC{$\Gamma \Rightarrow \,$}
\RightLabel{\scriptsize $(0w)$}
\UnaryInfC{$\Gamma \Rightarrow 0$}
\DisplayProof
\end{tabular}
\end{center}
Define $\mathcal{S}_\pi=\mathcal{S}_{\pi_1}$. By the induction hypothesis, we have
$\mathcal{S}_{\pi_1} \vdash_{\LKajab} \Gamma^t \Rightarrow \,$.
For $(i)$, by the rule $(Lw)$ in $\LKajab$, we obtain:
\begin{center}
\begin{tabular}{c}
\AxiomC{$\Gamma^t \Rightarrow \,$}
\UnaryInfC{$\Gamma^t \Rightarrow 0^t$}
\DisplayProof
\end{tabular}
\end{center}
Property $(ii)$ follows immediately from the induction hypothesis. For $(iii)$, it suffices to define
$\sigma_{\pi}=\sigma_{\pi_1}$.

\vspace{4pt}
\noindent $\bullet$ If the last rule in $\pi$ is $(R*)$, then $\pi$ is of the form:
\begin{center}
\begin{tabular}{c}
\AxiomC{$\pi_1$}
\noLine
\UnaryInfC{$\Gamma_1 \Rightarrow A$}

\AxiomC{$\pi_2$}
\noLine
\UnaryInfC{$\Gamma_2 \Rightarrow B$}

\BinaryInfC{$\Gamma_1,\Gamma_2 \Rightarrow A*B$}
\DisplayProof
\end{tabular}
\end{center}
Define
$\mathcal{S}_\pi=\mathcal{S}_{\pi_1}\cup\mathcal{S}_{\pi_2}\cup
\{\langle A\rangle,\langle B\rangle\Rightarrow\langle A*B\rangle\}$.
By the induction hypothesis, we have
$\mathcal{S}_{\pi_1}\vdash_{\LKajab}\Gamma_1^t\Rightarrow A^t$
and
$\mathcal{S}_{\pi_2}\vdash_{\LKajab}\Gamma_2^t\Rightarrow B^t$.
For $(i)$, by the rule $(R*)$, we obtain:
\begin{center}
\begin{tabular}{c}
\AxiomC{$\Gamma_1^t \Rightarrow A^t$}
\AxiomC{$\Gamma_2^t \Rightarrow B^t$}
\BinaryInfC{$\Gamma_1^t,\Gamma_2^t \Rightarrow A^t*B^t$}
\DisplayProof
\end{tabular}
\end{center}
As $\LKajab\vdash A^t\Rightarrow\langle A\rangle$ and
$\LKajab\vdash B^t\Rightarrow\langle B\rangle$, we have
$\LKajab\vdash A^t*B^t\Rightarrow\langle A\rangle*\langle B\rangle$.
Therefore, using
$\langle A\rangle,\langle B\rangle\Rightarrow\langle A*B\rangle$
from $\mathcal{S}_\pi$, we obtain
$\mathcal{S}_\pi\vdash_{\LKajab}\Gamma_1^t,\Gamma_2^t\Rightarrow\langle A*B\rangle$.
Finally, we obtain
$\mathcal{S}_\pi\vdash_G\Gamma_1^t,\Gamma_2^t\Rightarrow(A*B)^t$.
Property $(ii)$ is immediate, since the sequent
$\langle A\rangle,\langle B\rangle\Rightarrow\langle A*B\rangle$ is Horn.
For $(iii)$, it suffices to define $\sigma_\pi$ in terms of
$\sigma_{\pi_1}$ and the canonical $G$-proof of
$\Rightarrow A*B\to A*B$.

\vspace{4pt}
\noindent $\bullet$ If the last rule in $\pi$ is $(R\wedge)$, then $\pi$ is of the form:
\begin{center}
\begin{tabular}{c}
\AxiomC{$\pi_1$}
\noLine
\UnaryInfC{$\Gamma \Rightarrow A$}

\AxiomC{$\pi_2$}
\noLine
\UnaryInfC{$\Gamma \Rightarrow B$}

\BinaryInfC{$\Gamma \Rightarrow A\wedge B$}
\DisplayProof
\end{tabular}
\end{center}
Define
$\mathcal{S}_\pi=\mathcal{S}_{\pi_1}\cup\mathcal{S}_{\pi_2}\cup
\{\langle A\rangle\wedge\langle B\rangle\Rightarrow\langle A\wedge B\rangle\}$.
By the induction hypothesis, we have
$\mathcal{S}_{\pi_1}\vdash_{\LKajab}\Gamma^t\Rightarrow A^t$
and
$\mathcal{S}_{\pi_2}\vdash_{\LKajab}\Gamma^t\Rightarrow B^t$.
For $(i)$, by the rule $(R\wedge)$, we obtain:
\begin{center}
\begin{tabular}{c}
\AxiomC{$\Gamma^t\Rightarrow A^t$}
\AxiomC{$\Gamma^t\Rightarrow B^t$}
\BinaryInfC{$\Gamma^t\Rightarrow A^t\wedge B^t$}
\DisplayProof
\end{tabular}
\end{center}
As $\LKajab\vdash A^t\Rightarrow\langle A\rangle$ and
$\LKajab\vdash B^t\Rightarrow\langle B\rangle$, we have
$\LKajab\vdash A^t\wedge B^t\Rightarrow
\langle A\rangle\wedge\langle B\rangle$.
Therefore, using
$\langle A\rangle\wedge\langle B\rangle\Rightarrow\langle A\wedge B\rangle$
from $\mathcal{S}_\pi$, we obtain
$\mathcal{S}_\pi\vdash_{\LKajab}\Gamma^t\Rightarrow\langle A\wedge B\rangle$.
Finally, we obtain
$\mathcal{S}_\pi\vdash_{\LKajab}\Gamma^t\Rightarrow(A\wedge B)^t$.
Property $(ii)$ is immediate, since the sequent
$\langle A\rangle\wedge\langle B\rangle\Rightarrow\langle A\wedge B\rangle$
is Horn. For $(iii)$, it suffices to define $\sigma_\pi$ in terms of
$\sigma_{\pi_1}$, $\sigma_{\pi_2}$, and the canonical $G$-proof of
$\Rightarrow A\wedge B\to A\wedge B$.

\vspace{4pt}
\noindent $\bullet$ The case of the rule $(R\vee)$ is similar to that of $(R\wedge)$, with either the Horn sequent $\langle A\rangle \Rightarrow \langle A\vee B\rangle$ or the Horn sequent $\langle B\rangle \Rightarrow \langle A\vee B\rangle$ added to $\mathcal{S}_{\pi_1}$.

\vspace{4pt}
\noindent $\bullet$ If the last rule in $\pi$ is $(R\to)$, then $\pi$ is of the form:
\begin{center}
\begin{tabular}{c}
\AxiomC{$\pi_1$}
\noLine
\UnaryInfC{$\Gamma,A \Rightarrow B$}
\RightLabel{\scriptsize $(R\to)$}
\UnaryInfC{$\Gamma \Rightarrow A\to B$}
\DisplayProof
\end{tabular}
\end{center}
Define
$\mathcal{S}_\pi=\mathcal{S}_{\pi_1}\cup
\{\{\langle\gamma\rangle\mid\gamma\in\Gamma\}\Rightarrow\langle A\to B\rangle\}$.
By the induction hypothesis, we have
$\mathcal{S}_{\pi_1}\vdash_{\LKajab}\Gamma^t,A^t\Rightarrow B^t$.
For $(i)$, by the rule $(R\to)$ in $\LKajab$, we obtain:
\begin{center}
\begin{tabular}{c}
\AxiomC{$\Gamma^t,A^t\Rightarrow B^t$}
\UnaryInfC{$\Gamma^t\Rightarrow A^t\to B^t$}
\DisplayProof
\end{tabular}
\end{center}
As $\LKajab\vdash\gamma^t\Rightarrow\langle\gamma\rangle$ for every $\gamma\in\Gamma$, we obtain
$\Gamma^t\Rightarrow\bigast_{\gamma\in\Gamma}\langle\gamma\rangle$.
Therefore, using
$\{\langle\gamma\rangle\mid\gamma\in\Gamma\}\Rightarrow\langle A\to B\rangle$
from $\mathcal{S}_\pi$, we obtain
$\mathcal{S}_\pi\vdash_{\LKajab}\Gamma^t\Rightarrow\langle A\to B\rangle$.
Finally, we obtain
$\mathcal{S}_\pi\vdash_{\LKajab}\Gamma^t\Rightarrow(A\to B)^t$.
Property $(ii)$ is immediate, since the sequent
$\{\langle\gamma\rangle\mid\gamma\in\Gamma\}\Rightarrow\langle A\to B\rangle$
is Horn. For $(iii)$, it suffices to define $\sigma_\pi$ in terms of
$\sigma_{\pi_1}$ and the $G$-proof of
$\Rightarrow\bigast\Gamma\to(A\to B)$ canonically constructed from the
$G$-proof $\pi$ of $\Gamma\Rightarrow A\to B$.

\vspace{4pt}
\noindent $\bullet$ If the last rule in $\pi$ is $(R!)$, then $\pi$ is of the form:
\begin{center}
\begin{tabular}{c}
\AxiomC{$\pi_1$}
\noLine
\UnaryInfC{$!\Gamma \Rightarrow A$}
\RightLabel{\scriptsize $(R!)$}
\UnaryInfC{$!\Gamma \Rightarrow !A$}
\DisplayProof
\end{tabular}
\end{center}
Define
\[
\mathcal{S}_\pi
=
\mathcal{S}_{\pi_1}
\cup
\left\{
\langle !\gamma\rangle \Rightarrow !\langle !\gamma\rangle
\mid \gamma\in\Gamma
\right\}
\cup
\left\{
\{\langle !\gamma\rangle\mid\gamma\in\Gamma\}
\Rightarrow \langle !A\rangle
\right\}.
\]
By the induction hypothesis, we have
$\mathcal{S}_{\pi_1}\vdash_{\LKajab}(!\Gamma)^t\Rightarrow A^t$.
Using $(L!)$ and then $(R!)$, we obtain
\[
\mathcal{S}_\pi\vdash_{\LKajab} !(!\Gamma)^t\Rightarrow !A^t.
\]
Now let $\gamma\in\Gamma$ be arbitrary. Clearly,
$\LKajab\vdash !\gamma^t\Rightarrow !!\gamma^t$, while
$
\mathcal{S}_\pi\vdash_{\LKajab}
\langle !\gamma\rangle\Rightarrow !\langle !\gamma\rangle$.
Thus,
$
\mathcal{S}_\pi\vdash_{\LKajab}
!\gamma^t\wedge\langle !\gamma\rangle
\Rightarrow
!!\gamma^t\wedge!\langle !\gamma\rangle.
$
As $!$ distributes over conjunction in $\LKajab$, we obtain
\[
\mathcal{S}_\pi\vdash_{\LKajab}
!\gamma^t\wedge\langle !\gamma\rangle
\Rightarrow
!(!\gamma^t\wedge\langle !\gamma\rangle),
\]
or, equivalently,
$
\mathcal{S}_\pi\vdash_{\LKajab}
(!\gamma)^t\Rightarrow !(!\gamma)^t.
$
Using cut together with
$\mathcal{S}_\pi\vdash_{\LKajab}!(!\Gamma)^t\Rightarrow !A^t$, we obtain
$
\mathcal{S}_\pi\vdash_{\LKajab}
(!\Gamma)^t\Rightarrow !A^t.
$
Finally, since
$\LKajab\vdash(!\gamma)^t\Rightarrow\langle !\gamma\rangle$,
using the sequent
$\{\langle !\gamma\rangle\mid\gamma\in\Gamma\}\Rightarrow\langle !A\rangle$
from $\mathcal{S}_\pi$, we obtain
$
\mathcal{S}_\pi\vdash_{\LKajab}
(!\Gamma)^t\Rightarrow\langle !A\rangle
$
and hence
$
\mathcal{S}_\pi\vdash_{\LKajab}
(!\Gamma)^t\Rightarrow(!A)^t.
$
Property $(ii)$ is immediate, since the sequents
$\langle !\gamma\rangle\Rightarrow!\langle !\gamma\rangle$
are of the allowed form, while
$\{\langle !\gamma\rangle\mid\gamma\in\Gamma\}\Rightarrow\langle !A\rangle$
is Horn. For $(iii)$, it suffices to define $\sigma_\pi$ in terms of
$\sigma_{\pi_1}$, the canonical $G$-proofs of
$\Rightarrow !\gamma\to!!\gamma$, and the canonical $G$-proof of
$\Rightarrow\bigast !\Gamma\to!A$ constructed from the $G$-proof $\pi$ of
$!\Gamma\Rightarrow!A$.

This completes the construction of $\mathcal{S}_{\pi}$ and $\sigma_{\pi}$ and establishes their required properties. For the size bound, first note that $\mathcal{S}_{\pi}$ is constructed by recursion on the structure of $\pi$. If $\pi$ is an axiom, then the size of $\mathcal{S}_{\pi}$ is constant. In each inductive step, $\mathcal{S}_{\pi}$ is defined as the union of $\mathcal{S}_{\pi_1}$ (resp.\ $\mathcal{S}_{\pi_1}\cup\mathcal{S}_{\pi_2}$) when $\pi$ is unary (resp.\ binary), together with some new sequents whose size is clearly bounded by the number of formulas in $\pi$, i.e., $O(l(\pi))$. As $\pi$ has at most $l(\pi)$ rules, we have
$|\mathcal{S}_{\pi}|\leq l(\pi)^{O(1)}$.
Notice that this size bound relies on the tree-likeness of $\pi$.

For $\sigma_{\pi}$, again, $\sigma_{\pi}$ is constructed by recursion on the structure of $\pi$. If $\pi$ is an axiom, then the size of $\sigma_{\pi}$ is constant. In each inductive step, $\sigma_{\pi}$ is defined by applying $(R\wedge)$ rules to $\sigma_{\pi_1}$ (resp.\ $\sigma_{\pi_1}$ and $\sigma_{\pi_2}$) when $\pi$ is unary (resp.\ binary), together with $G$-proofs of the standard translations of some new sequents. The latter are either proofs of constant size or $G$-proofs constructed from $\pi$ by adding $O(l(\pi))$ new rules, each containing $O(l(\pi))$ formulas. Hence, at each step, we add at most $l(\pi)^{O(1)}$ formulas. Since $\pi$ has at most $l(\pi)$ rules and is tree-like, we obtain
$
l(\sigma_{\pi})\leq l(\pi)^{O(1)}.
$
\end{proof}

\subsection{Feasible Disjunction Property}

In this subsection, we prove a form of the feasible disjunction property in which the disjuncts are atomic and the assumptions are Horn. Later, we use Theorem~\ref{MainTheorem} to lift this special case to the more general form of the feasible disjunction property that we seek.

\begin{theorem}(Unit Propagation)\label{UnitPropagation}
For any set $\mathcal{S}$ of Horn sequents over $\mathcal{L}_b$ and any two atoms $p$ and $q$, if $\mathcal{S} \vdash_{\LKb} \, \Rightarrow p \vee q$, then there is $\tau$ such that $l(\tau) \leq |\mathcal{S}|^{O(1)}$ and either $\mathcal{S} \vdash_{\mathbf{FL_e}}^{\tau} \, \Rightarrow p$ or $\mathcal{S} \vdash_{\mathbf{FL_e}}^{\tau} \, \Rightarrow q$. 
\end{theorem}
\begin{proof}
Call any sequent in the form $(\, \Rightarrow u)$, where $u$ is an atom, a \emph{unit}. Now, define two sequences $\{\mathcal{S}_i\}_{i=0}^{\infty}$ and $\{\mathcal{U}_i\}_{i=0}^{\infty}$ of sequents in the following inductive way. Set $\mathcal{S}_0=\mathcal{S}$ and $\mathcal{U}_0=\varnothing$. Then, define $\mathcal{S}_{i+1}$ and $\mathcal{U}_{i+1}$ by applying the following process.
If there is no unit in $\mathcal{S}_i-\mathcal{U}_i$, set $\mathcal{S}_{i+1}=\mathcal{S}_{i}$ and $\mathcal{U}_{i+1}=\mathcal{U}_i$. Otherwise, pick a unit $(\, \Rightarrow u)$ arbitrarily in $\mathcal{S}_i-\mathcal{U}_i$. Now, for any sequent $S \in \mathcal{S}_i$ other than $(\, \Rightarrow u)$, apply the following simplifications to modify $\mathcal{S}_i$:
\begin{itemize}
\item[$(i)$] 
If $u$ does not appear in $S$, do not change $S$.
\item[$(ii)$] 
If $S$ is in the form $C_1, \cdots, C_m \Rightarrow u$, erase $S$ from $\mathcal{S}_i$.
\item[$(iii)$] 
If $S$ is in the form $C_1, \cdots, C_m \Rightarrow r$, where $r \neq u$, then substitute all occurrences of $u$ in the $C_i$'s by $1$.
\end{itemize}
When this process has been applied to all sequents of $\mathcal{S}_i$ except for $(\, \Rightarrow u)$, eliminate any formula that is just a conjunction of $1$'s from the antecedents of the resulting sequents, and set $\mathcal{S}_{i+1}$ to be the result. For $\mathcal{U}_{i+1}$, simply set it to $\mathcal{U}_i \cup \{\, \Rightarrow u\}$.

Here are some properties of these two sequences. First,
$\mathcal{U}_i \subseteq \mathcal{S}_i$ for any $i \in \mathbb{N}$ and each $\mathcal{S}_i$ is a set of Horn sequents satisfying $|\mathcal{S}_i|\subseteq |\mathcal{S}|$. To prove the first part, we use induction on $i \in \mathbb{N}$. For $i=0$, the claim is obvious. For the induction step, if there is no unit in $\mathcal{S}_i-\mathcal{U}_i$, then $\mathcal{S}_{i+1}=\mathcal{S}_{i}$ and $\mathcal{U}_{i+1}=\mathcal{U}_i$, and the claim follows from the induction hypothesis for $i$. Otherwise, let $(\, \Rightarrow u)$ be the chosen unit in $\mathcal{S}_i-\mathcal{U}_i$. By definition, $(\, \Rightarrow u)$ belongs to $\mathcal{S}_i$ and remains unchanged throughout the modifications. Hence, $(\, \Rightarrow u) \in \mathcal{S}_{i+1}$. For any other unit $(\, \Rightarrow v) \in \mathcal{U}_{i+1}$, it must belong to $\mathcal{U}_i$ and hence, by the induction hypothesis, $(\, \Rightarrow v) \in \mathcal{S}_i$. As $v \neq u$, the sequent $(\, \Rightarrow v)$ remains intact throughout the modifications, which only affect occurrences of $u$. Therefore, $(\, \Rightarrow v) \in \mathcal{S}_{i+1}$. For the second part, being Horn is clear from the modifications. For the bound, it is enough to note that $|\mathcal{S}_{i+1}| \leq |\mathcal{S}_i|$, for any $i\in \mathbb{N}$. 

Second, by an easy induction, one can show that, for any $i\in\mathbb{N}$, if an atom occurs in $\mathcal{U}_i$, then it appears only in a unit in $\mathcal{S}_i$. For $i=0$, there is nothing to prove. For the induction step, if there is no unit in $\mathcal{S}_i-\mathcal{U}_i$, then $\mathcal{S}_{i+1}=\mathcal{S}_i$ and $\mathcal{U}_{i+1}=\mathcal{U}_i$, and the claim follows from the induction hypothesis for $i$. Otherwise, let $(\,\Rightarrow u)$ be the chosen unit in $\mathcal{S}_i-\mathcal{U}_i$, and let $(\,\Rightarrow v)$ be a unit in $\mathcal{U}_{i+1}$. There are two cases. If $v=u$, then, since all occurrences of $u$ in $\mathcal{S}_i$ except for $(\,\Rightarrow u)$ are eliminated when defining $\mathcal{S}_{i+1}$, the only occurrence of $u$ in $\mathcal{S}_{i+1}$ is in the unit $(\,\Rightarrow u)$. If $v\neq u$, then $(\,\Rightarrow v)\in\mathcal{U}_i$, and by the induction hypothesis, $v$ occurs only in units of $\mathcal{S}_i$. Since the modifications only remove sequents or replace occurrences of $u$, the claim follows.

Third, we claim 
\[
\LKb \vdash
\bigwedge_{S\in\mathcal{S}_{i+1}} I(S)
\Leftrightarrow
\bigwedge_{S\in\mathcal{S}_{i}} I(S),
\]
for any $i\in\mathbb{N}$. This follows directly from the following
$\LKb$-equivalences:
\begin{itemize}
    \item
    $u\wedge(A\to u)\Leftrightarrow u$;
    
    \item
    $u\wedge((A\wedge u)*B\to r)
    \Leftrightarrow
    u\wedge((A\wedge 1)*B\to r)$;
    
    \item
    $(A*B\to r)\Leftrightarrow(B\to r)$,
    where $A$ is a conjunction of $1$'s.
\end{itemize}
Notice that we are working in the system $\LKb$, where all structural rules
are available. Hence, these are simply the usual classical equivalences
expressed in the substructural language.

Fourth, for any $i\in\mathbb{N}$, there is an $\mathbf{FL_e}$-proof with
$|\mathcal{S}|^{O(1)}$ many lines for
\[
\Rightarrow\bigwedge_{S\in\mathcal{S}_i} I(S)
\vdash_{\mathbf{FL_e}}
\Rightarrow\bigwedge_{S\in\mathcal{S}_{i+1}} I(S).
\]
There are two cases to consider. If there is no unit in
$\mathcal{S}_i-\mathcal{U}_i$, then
$\mathcal{S}_{i+1}=\mathcal{S}_i$, and hence there is nothing to prove.
Otherwise, let $(\,\Rightarrow u)\in\mathcal{S}_i-\mathcal{U}_i$ be the
chosen unit and let $T$ be an arbitrary sequent in $\mathcal{S}_{i+1}$.
Assume that $T$ is obtained by modifying some $R\in\mathcal{S}_i$.
Since conjunctions of $1$'s are equivalent to $1$ and can be eliminated
from antecedents in $\mathbf{FL_e}$, without loss of generality, we may
assume that $R$ is the result of either $(i)$ or $(iii)$. Note that $(ii)$
does not apply, as it eliminates sequents.
For $(i)$, we have $T=R\in\mathcal{S}_i$, and hence
$
\Rightarrow\bigwedge_{S\in\mathcal{S}_i} I(S)
\vdash_{\mathbf{FL_e}}
\Rightarrow I(T).
$
For $(iii)$, let
\[
R=(\Gamma,A_1\wedge u,\ldots,A_m\wedge u\Rightarrow r),
\]
where $r\neq u$ and $\Gamma$ contains no occurrence of $u$. Then
\[
T=(\Gamma,A_1\wedge1,\ldots,A_m\wedge1\Rightarrow r).
\]
Now, observe that
\[
\{(\,\Rightarrow u),
(\Gamma,\{A_j\wedge u\}_{j=1}^m\Rightarrow r)\}
\vdash_{\mathbf{FL_e}}
(\Gamma,\{A_j\wedge1\}_{j=1}^m\Rightarrow r).
\]
Therefore, since $(\,\Rightarrow u)\in\mathcal{S}_i$, we obtain
$
\Rightarrow\bigwedge_{S\in\mathcal{S}_i} I(S)
\vdash_{\mathbf{FL_e}}
\Rightarrow I(T).
$
For the number of lines in the proof, first note that for each
$T\in\mathcal{S}_{i+1}$ in case $(i)$, the number of lines is clearly
$O(|\mathcal{S}_i|)\leq O(|\mathcal{S}|)$, by the first property. In case
$(iii)$, the number of lines is
$|\mathcal{S}_i|^{O(1)}\leq|\mathcal{S}|^{O(1)}$, again by the first
property. Since the number of sequents $T$ is
$O(|\mathcal{S}_{i+1}|)\leq O(|\mathcal{S}|)$, we obtain an
$\mathbf{FL_e}$-proof with $|\mathcal{S}|^{O(1)}$ many lines witnessing
\[
\Rightarrow\bigwedge_{S\in\mathcal{S}_i} I(S)
\vdash_{\mathbf{FL_e}}
\Rightarrow\bigwedge_{S\in\mathcal{S}_{i+1}} I(S).
\]
Now, we use these four properties to prove the claim. First, by the first
property, we have $\mathcal{U}_i\subseteq\mathcal{S}_i$ and
$|\mathcal{S}_i|\leq|\mathcal{S}|$ for any $i\in\mathbb{N}$. We claim that
at some point $N\in\mathbb{N}$, there is no unit in
$\mathcal{S}_N-\mathcal{U}_N$. Otherwise, the sequence
$\{\mathcal{U}_i\}_{i\in\mathbb{N}}$ would be a strictly increasing
sequence of sets whose sizes are bounded by $|\mathcal{S}|$, which is
impossible. Now, let $N$ be a stage at which the two sequences
$\{\mathcal{S}_i\}_{i\in\mathbb{N}}$ and
$\{\mathcal{U}_i\}_{i\in\mathbb{N}}$ stabilize. Notice that $N$ is bounded
by the number of atoms occurring in $\mathcal{S}$, and hence
$N\leq|\mathcal{S}|$.

We claim that either $(\, \Rightarrow p)$ or $(\, \Rightarrow q)$ appears in $\mathcal{S}_N$. Let us assume otherwise. As $\mathcal{S} \vdash_{\LKb} \, \Rightarrow p \vee q$, we have $\LKb \vdash \bigwedge_{S \in \mathcal{S}} I(S) \Rightarrow p \vee q$. Now, by the third property, we reach $\LKb \vdash \bigwedge_{S \in \mathcal{S}_N} I(S) \Rightarrow p \vee q$, which is equivalent to $\mathcal{S}_N \vdash_{\LKb} \, \Rightarrow p \vee q$. Therefore, any Boolean valuation satisfying $\mathcal{S}_N$ must validate $(\, \Rightarrow p \vee q)$.
Define the valuation $V$ by setting $V(u)=1$ if and only if $(\, \Rightarrow u) \in \mathcal{U}_N$. We claim that $V(S)=1$ for any $S \in \mathcal{S}_N$. If $S$ is a unit, then $S=(\, \Rightarrow u)$. As there is no unit in $\mathcal{S}_N-\mathcal{U}_N$, we must have $(\, \Rightarrow u) \in \mathcal{U}_N$. Therefore, $V(u)=1$, which implies $V(S)=1$. If $S$ is not a unit, then $S=(C_1, \cdots, C_m \Rightarrow r)$, where $m$ is non-zero. Therefore, at least one atom $s$ appears in the antecedent of $S$. By the second property, since $s$ appears in a non-unit sequent $S \in \mathcal{S}_N$, we must have $(\, \Rightarrow s) \notin \mathcal{U}_N$. Hence, $V(s)=0$, which implies $V(S)=1$. Therefore, $V$ validates all sequents in $\mathcal{S}_N$.

Now, by the first property, $\mathcal{U}_i \subseteq \mathcal{S}_i$ for any $i \in \mathbb{N}$. Hence, by the assumption that $(\, \Rightarrow p)$ and $(\, \Rightarrow q)$ are not in $\mathcal{S}_N$, they are also outside of $\mathcal{U}_N$. Thus, $V(p)=V(q)=0$. Therefore, $V$ does not validate $(\, \Rightarrow p \vee q)$, which is a contradiction with $\mathcal{S}_N \vdash_{\LKb} \, \Rightarrow p \vee q$. Therefore, either $(\, \Rightarrow p)$ or $(\, \Rightarrow q)$ appears in $\mathcal{S}_N$. Hence, either $\Rightarrow \bigwedge_{S \in \mathcal{S}_N} I(S) \vdash_{\mathbf{FL_e}} \, \Rightarrow p$ or $\Rightarrow \bigwedge_{S \in \mathcal{S}_N} I(S) \vdash_{\mathbf{FL_e}} \, \Rightarrow q$. Note that the proofs have $O(|\mathcal{S}_N|)\leq O(|\mathcal{S}|)$ many lines.
By the fourth property, for any $i \in \mathbb{N}$, we have 
\[
\, \Rightarrow \bigwedge_{S \in \mathcal{S}_i} I(S) \vdash_{\mathbf{FL_e}} \, \Rightarrow \bigwedge_{S \in \mathcal{S}_{i+1}} I(S)
\]
with a proof containing $|\mathcal{S}|^{O(1)}$ many lines. As $N \leq |\mathcal{S}|$, we obtain a proof of either $\, \Rightarrow \bigwedge_{S \in \mathcal{S}} I(S) \vdash_{\mathbf{FL_e}} \, \Rightarrow p$ or $\, \Rightarrow \bigwedge_{S \in \mathcal{S}} I(S) \vdash_{\mathbf{FL_e}} \, \Rightarrow q$ with $|\mathcal{S}|^{O(1)}$ many lines. Therefore, there is an $\mathbf{FL_e}$-proof $\tau$ of $\mathcal{S} \vdash_{\mathbf{FL_e}} \, \Rightarrow p$ or $\mathcal{S} \vdash_{\mathbf{FL_e}} \, \Rightarrow q$ such that $l(\tau) \leq |\mathcal{S}|^{O(1)}$.
\end{proof}

Finally, we have developed enough machinery to prove the main theorem of this section. Define the forgetful function $f:\mathcal{L}_! \to \mathcal{L}_b$ by $f(p)=p$ for any atom $p$, $f(c)=c$ for any constant $c \in \mathcal{L}_!$, $f(A \circ B)=f(A) \circ f(B)$ for any $\circ \in \{\wedge, \vee, \to, *\}$, and $f(!A)=A$. It is clear that $f$ maps $\LKajab$-provability to $\LKb$-provability for $!$-free sequents.

\begin{theorem} \label{FeasibleDP}
Any structural calculus $G$ such that $L_G \subseteq \mathsf{LK_!}$ has the feasible disjunction property. 
\end{theorem}
\begin{proof}
Let $\pi$ be a $G$-proof of $\Rightarrow A \vee B$. By Theorem \ref{MainTheorem}, we obtain a $G$-proof $\sigma_{\pi}$ and a set $\mathcal{S}_{\pi}$ consisting of Horn sequents or sequents of the form $\langle !C \rangle \Rightarrow !\langle !C \rangle$ such that $|\mathcal{S}_{\pi}| \leq l(\pi)^{O(1)}$ and
\begin{center}
    $\mathcal{S}_{\pi} \vdash_{\LKajab} \, \Rightarrow (A \vee B)^t
    \quad \text{and} \quad
    G \vdash^{\sigma_{\pi}} (\, \Rightarrow \bigwedge I(\mathcal{S}_{\pi}^s))$.
\end{center}
Set $\mathcal{T}$ to be the subset of $\mathcal{S}_{\pi}$ consisting only of Horn sequents. By Definition \ref{Translation}, $\LKajab \vdash (A \vee B)^t \Rightarrow A^t \vee B^t$. Hence, by Lemma \ref{TranslationAndAtoms}, we obtain
$\mathcal{S}_{\pi} \vdash_{\LKajab} \, \Rightarrow \langle A \rangle \vee \langle B \rangle$.
Now, applying the forgetful function $f$, note that
$f(\langle !C \rangle \Rightarrow !\langle !C \rangle)
= \langle !C \rangle \Rightarrow \langle !C \rangle$,
which is $\LKb$-provable. Since Horn sequents are $!$-free, we therefore obtain
$\mathcal{T} \vdash_{\LKb} \, \Rightarrow \langle A \rangle \vee \langle B \rangle$.
Then, by Theorem \ref{UnitPropagation}, there is a $\mathbf{FL_e}$-proof $\tau$ with
$l(\tau) \leq |\mathcal{T}| \leq l(\pi)^{O(1)}$ such that either
    $\mathcal{T} \vdash_{\mathbf{FL_e}}^{\tau} \, \Rightarrow \langle A \rangle$
   or
    $\mathcal{T} \vdash_{\mathbf{FL_e}}^{\tau} \, \Rightarrow \langle B \rangle$.
Using the standard substitution and the $G$-proof $\sigma_{\pi}$ of
$(\, \Rightarrow \bigwedge I(\mathcal{T})^s)$ and the facts $\mathbf{FL_e}\subseteq G$ and $|\mathcal{T}|\leq |\mathcal{S}_{\pi}| \leq l(\pi)^{O(1)}$, we obtain a $G$-proof of $(\, \Rightarrow A)$ or $(\, \Rightarrow B)$ with $|\pi|^{O(1)}$ number of lines.
\end{proof}

\begin{theorem}\label{Thm: MainFrege}
For any structural logic $L \subseteq \mathsf{LK_!}$, the system $L$-Frege enjoys feasible disjunction property.
\end{theorem}

\section{Exponential Separations}

In this section, we prove the promised exponential separations. First, we recall the known lower bounds for substructural systems.

\subsection{Substructural Clique-Color Formula}

Let $n,k,m\geq 1$. We encode graphs on the vertex set $[n]=\{1,\dots,n\}$ using propositional variables
$
\{p_{i,j}\mid i\neq j\in[n]\}.
$
An assignment to these variables determines a graph on $[n]$, where $p_{i,j}$ is assigned the value $1$ precisely when the vertices $i$ and $j$ are joined by an edge. Recall that a graph contains a $k$-\emph{clique} if it admits a complete subgraph on $k$ vertices, and that it is \emph{$m$-colorable} if its vertices can be colored with $m$ colors so that adjacent vertices receive different colors.

To formalize the existence of a $k$-clique, we introduce variables
$
\{r_{u,i}\mid u\in[k],\, i\in[n]\}.
$
Intuitively, the variables $r_{u,i}$ describe a map from $[k]$ into $[n]$, with $r_{u,i}=1$ expressing that $u$ is sent to $i$. Consider the collection of formulas:

\begin{itemize}
\item[$\bullet$]
$\displaystyle \bigvee_{i\in[n]} r_{u,i}$, for each $u\in[k]$,
\item[$\bullet$]
$\displaystyle \neg r_{u,i}\vee \neg r_{u,j}$, for every $u\in[k]$ and distinct $i,j\in[n]$,
\item[$\bullet$]
$\displaystyle \neg r_{u,i}\vee \neg r_{v,i}$, for all distinct $u,v\in[k]$ and $i\in[n]$,
\item[$\bullet$]
$\displaystyle \neg r_{u,i}\vee \neg r_{v,j}\vee p_{i,j}$, for all distinct $u,v\in[k]$ and distinct $i,j\in[n]$.
\end{itemize}
We write $Clique^k_n(\bar p,\bar r)$ for the conjunction of all these clauses. The first three groups ensure that the variables $\bar r$ encode an injective map from $[k]$ to $[n]$, while the last group guarantees that any two selected vertices are adjacent. Consequently, $Clique^k_n(\bar p,\bar r)$ expresses that the graph represented by $\bar p$ contains a $k$-clique. 

Next we encode colorings. For this purpose, introduce variables
$
\{s_{i,a}\mid i\in[n],\, a\in[m]\},
$
where $s_{i,a}=1$ is intended to mean that vertex $i$ receives color $a$. Define $Color^m_n(\bar p,\bar s)$ as the conjunction of the following formulas:

\begin{itemize}
\item[$\bullet$]
$\displaystyle \bigvee_{a\in[m]} s_{i,a}$, for every $i\in[n]$,
\item[$\bullet$]
$\displaystyle \neg s_{i,a}\vee \neg s_{i,b}$, for all distinct $a,b\in[m]$ and $i\in[n]$,
\item[$\bullet$]
$\displaystyle \neg s_{i,a}\vee \neg s_{j,a}\vee \neg p_{i,j}$, for all $a\in[m]$ and distinct $i,j\in[n]$.
\end{itemize}
The first two groups of clauses ensure that each vertex receives exactly one color, whereas the third forbids adjacent vertices from sharing a color. Thus $Color^m_n(\bar p,\bar s)$ expresses that the graph encoded by $\bar p$ admits a proper $m$-coloring.

Since a graph containing a $(k+1)$-clique cannot be $k$-colored, the formula
\[
Clique^{k+1}_n(\bar{p}, \bar{r}) \to \neg Color^{k}_n(\bar{p}, \bar{s})
\]
is a classical tautology. Hrubeš \cite{Hrubes} modified this formula to obtain an intuitionistic tautology and proved that it requires exponential many lines proofs in $\mathbf{LJ}$, and hence in the $\mathsf{IPC}$-Frege system.

\begin{theorem}\label{Hrubes} \cite{Hrubes}
Let $\bar{p}=\{p_{i,j} \mid 1 \leq i \neq j \leq n\}$, $\bar{q}=\{q_{i,j} \mid 1 \leq i \neq j \leq n\}$, $\bar{r}$, and $\bar{s}$ be pairwise disjoint sets of variables, and let $k=\lfloor\sqrt{n}\rfloor$. Then the formulas
\[
\Theta^{\bot}_n := \bigwedge_{1 \leq i \neq j \leq n} (p_{i,j} \vee q_{i,j}) \to \neg Color^{k}_n(\bar{p}, \bar{s}) \vee \neg Clique^{k+1}_n(\neg \bar{q}, \bar{r})
\]
are intuitionistic tautologies, and every $\LJ$-proof, equivalently every $\mathsf{IPC}$-Frege proof, of $\Theta^{\bot}_n$ contains at least $2^{\Omega(n^{1/4})}$ lines.
\end{theorem}

A negation-free variant of these encodings was subsequently proposed by Jeřábek~\cite{jerabek}. The idea is to replace negative literals by fresh propositional variables. More precisely, for every variable $s_{i,l}$ and $r_{i,l}$, one introduces new variables $s'_{i,l}$ and $r'_{i,l}$ intended to represent $\neg s_{i,l}$ and $\neg r_{i,l}$, respectively:

\begin{definition} \cite[Definition 6.28]{jerabek} \label{def: alpha beta}
For $k \leq n$, let
\[
\alpha^k_n (\bar{p}, \bar{s}, \bar{s'}) := \bigvee_{i\in [n]} \bigwedge_{l \in [k]} s'_{i,l} \;\;\vee\;\; \bigvee_{i\neq j \in [n]} \bigvee_{l \in [k]} (s_{i,l} \wedge s_{j,l} \wedge p_{i,j}),
\]
\[
\beta^k_n (\bar{q}, \bar{r}, \bar{r'}) := \bigvee_{l \in [k]} \bigwedge_{i \in [n]} r'_{i,l} \;\;\vee\;\; \bigvee_{i \neq j \in [n]} \bigvee_{l \neq m \in [k]} (r_{i,l} \wedge r_{j,m} \wedge q_{i,j}).
\]
Define the negation-free version of Hrube\v{s}'s formulas by
\[
\Theta_{n,k}:=  \big(\bigwedge_{i,j} (p_{i,j} \vee q_{i,j})\big) \; \to 
\]
\[
[
(\bigwedge_{i,l}(s_{i,l} \vee s'_{i,l}) \to \alpha^k_n (\bar{p}, \bar{s}, \bar{s'})) \vee  (\bigwedge_{i,l}(r_{i,l} \vee r'_{i,l}) \to \beta^{k+1}_n (\bar{q}, \bar{r}, \bar{r'}))].
\]
% \begin{align*}
% \Theta_{n,k}:= & \big(\bigwedge_{i,j} (p_{i,j} \vee q_{i,j})\big) \; \to  \\
% &[
% (\bigwedge_{i,l}(s_{i,l} \vee s'_{i,l}) \to \alpha^k_n (\bar{p}, \bar{s}, \bar{s'})) \vee  \\
% &(\bigwedge_{i,l}(r_{i,l} \vee r'_{i,l}) \to \beta^{k+1}_n (\bar{q}, \bar{r}, \bar{r'}))].
% \end{align*}
Note that
\[
 Color^k_n(\bar{p}, \bar{s}) = \neg \alpha^k_n(\bar{p}, \bar{s}, \neg \bar{s}) \qquad 
Clique^k_n(\bar{p}, \bar{r}) = \neg \beta^k_n(\neg \bar{p}, \bar{r}, \neg \bar{r}).
\] 
\end{definition}
%The lower bound of Theorem~\ref{Hrubes} also holds for~$\Theta_n$ \cite{jerabek}. 

\begin{theorem}(\cite[Theorem 6.37]{jerabek}) \label{JerabekAsli}
The formulas $\Theta_n := \Theta_{n,\lfloor \sqrt{n} \rfloor}$ are intuitionistic tautologies and every $\LJ$-proof, equivalently every $\mathsf{IPC}$-Frege proof, of $\Theta^{\bot}_n$ contains at least $2^{\Omega(n^{1/4})}$ lines.
\end{theorem}

Jalali~\cite{Raheleh} observed that a minor adjustment of $\Theta_{n, k}$ yields a formula provable in $\mathsf{FL_e}$. The resulting formula is
\[
\Theta^{*}_{n,k}:= [\bigast_{i,j} ((p_{i,j} \wedge 1) \vee (q_{i,j} \wedge 1))] \; \to 
\]
\[
\big([\bigast_{i,l} ((s_{i,l} \wedge 1) \vee (s'_{i,l} \wedge 1)) \to \alpha^k_n (\bar{p}, \bar{s}, \bar{s'})] 
\vee [\bigast_{i,l}((r_{i,l} \wedge 1) \vee  (r'_{i,l} \wedge 1)) \to \beta^{k+1}_n (\bar{q}, \bar{r}, \bar{r'})]\big).
\]
As in the previous construction, one sets
$\Theta^*_n:=\Theta^*_{n,\lfloor\sqrt n\rfloor}$.
The main lower-bound result established in~\cite{Raheleh} is the following.

\begin{theorem}[{\cite[Corollary 5.4]{Raheleh}}]
\label{Thm: exp lower bound FLe}
The formulas $\Theta^*_n$ are derivable in $\BFL$. Moreover, every $\LJajab$-proof, equivalently every $\mathsf{LJ_!}$-Frege proof, of $\Theta^{\bot}_n$ contains at least $2^{\Omega(n^{1/4})}$ lines.
\end{theorem}

\subsection{Separation Results}

In this subsection, we finally use the feasible disjunction property proved above to establish the promised exponential gaps.

\begin{definition}
A proof system $P$ \emph{feasibly admits disjunction introduction} if, for any formulas $A_0$ and $A_1$, any $i \in \{0,1\}$, and any $P$-proof $\pi$ of $A_i$, there is a $P$-proof $\sigma$ of $A_0 \vee A_1$ such that $|\sigma| \leq (|\pi|+|A_0|+|A_1|)^{O(1)}$.
\end{definition}

\begin{example}
For any $X \subseteq \{c,w\}$, the sequent calculi $\mathbf{FL_e}X$, $\mathbf{IMALL}X$, $\mathbf{ILL}X$, and their cut-free versions feasibly admit disjunction introduction. The same holds for the $L$-Frege system for any logic $L \supseteq \mathsf{FL_e}$.
\end{example}

The following is the main separation result we prove. We then present some specific cases as interesting corollaries.

\begin{definition}
Let $L$ and $M$ be logics, $Q$ be a proof system for $M$ and $P$ be either a sequent calculus or a Frege system for $L$. We say that $P$ is exponentially weaker than $Q$ if there is a sequence $\{A_n\}_{n \in \mathbb{N}}$ of $\mathsf{FL_e}$-provable formulas such that $A_n$ has a $Q$-proof of size $n^{O(1)}$, while any $P$-proof of $A_n$ has $2^{n^{\Omega(1)}}$ many lines.
\end{definition}

\begin{theorem}\label{thm: MainLowerBound}
Let $L \subseteq \mathsf{LJ_!}$ be a structural logic, $\mathsf{FL_e} \subseteq M \nsubseteq L$ be a logic, and $P$ be a proof system for $M$ that feasibly admits disjunction introduction. Then, $L$-Frege is exponentially weaker than $P$.
\end{theorem}
\begin{proof}
Using $M \nsubseteq L$, choose a formula $B \in M-L$. As $B \in M$, the formula $B$ has a $P$-proof, and since $B$ is fixed, the size of this proof is constant. First, note that the formula $\Theta^*_n \vee B$ is provable in $\mathsf{FL_e}$. Moreover, since $P$ feasibly admits disjunction introduction and $\Theta^*_n$ has size $n^{O(1)}$, the formula $\Theta^*_n \vee B$ has a $P$-proof of size $n^{O(1)}$.

For the lower bound, let $\pi$ be a $L$-Frege proof of $\Theta^{*}_n \vee B$. By the feasible disjunction property of $G$, Theorem \ref{Thm: MainFrege}, there is either a $L$-Frege proof of $\Theta^{*}_n$ or $B$ with $l(\pi)^{O(1)}$ many lines. As $B \notin L$, the latter case is impossible. Therefore, $ \Theta^{*}_n$ has an $L$-Frege proof with $l(\pi)^{O(1)}$ many lines. Since $L \subseteq \mathsf{LJ_!}$, we have $L-\mathbf{F} \leq_l \mathsf{LJ_!}-\mathbf{F}$, and hence $\Theta^{*}_n$ has an $\mathsf{LJ_!}$-Frege proof with $l(\pi)^{O(1)}$ many lines. By Theorem \ref{Thm: exp lower bound FLe}, we reach $l(\pi) \geq 2^{n^{\Omega(1)}}$.
\end{proof}

\begin{corollary}\label{cor: Main}
Let $L \subseteq \mathsf{LJ_!}$ be a structural logic and $M \supset L$ be a logic. Then $L$-Frege is exponentially weaker than $M$-Frege. In particular, this holds for any $L \in \{\mathsf{FL_e}X, \mathsf{IMALL}X, \mathsf{ILL}X\}$, where $X$ is a set of rules from Example \ref{exam StructuralRules}.
\end{corollary}
\begin{proof}
It is a direct consequence of Theorem \ref{thm: MainLowerBound} and the fact that, for any $M \supseteq \mathsf{FL_e}$, the system $M$-Frege feasibly admits disjunction introduction.
\end{proof}

Corollary \ref{cor: Main} shows that, over any structural logic below intuitionistic logic, extending the logic can yield an exponential reduction in the number of lines in Frege proofs, even for formulas already provable in $\mathsf{FL_e}$.

\begin{corollary}
Let $X, Y \subseteq \{w,c\}$ with $Y \nsubseteq X$. Then $\mathbf{FL_e}X$, and hence its cut-free version, is exponentially weaker than cut-free $\mathbf{FL_e}Y$. The same holds with $\mathbf{IMALL}$ or $\mathbf{ILL}$ in place of $\mathbf{FL_e}$.
\end{corollary}

Interestingly, this shows that $\mathbf{FL_{e}}$ is exponentially weaker than cut-free $\mathbf{FL_{ew}}$, that is, weakening can be exponentially more powerful than cut, even for $\mathbf{FL_e}$-provable formulas. A similar phenomenon occurs when adding contraction or weakening to any $\mathbf{FL_e}X$ that does not already contain the corresponding structural rule. \\

\noindent \textbf{Acknowledgment.} We wish to thank the financial support of the Dutch Research Council (NWO) project OCENW.M.22.258.

\bibliographystyle{plain}
\bibliography{Ref}

@article{Raheleh,
title = {Proof complexity of substructural logics},
journal = {Annals of Pure and Applied Logic},
volume = {172},
number = {7},
pages = {102972},
year = {2021},
issn = {0168-0072},
doi = {https://doi.org/10.1016/j.apal.2021.102972},
url = {https://www.sciencedirect.com/science/article/pii/S0168007221000300},
author = {Raheleh Jalali}
}

@article{jerabek,
  title={Substitution Frege and extended Frege proof systems in non-classical logics},
  author={Je{\v{r}}{\'a}bek, Emil},
  journal={Annals of Pure and Applied Logic},
  volume={159},
  number={1-2},
  pages={1--48},
  year={2009},
  publisher={Elsevier}
}

@article{Cook,
  title={The relative efficiency of propositional proof systems},
  author={Cook, Stephen A and Reckhow, Robert A},
  journal={The journal of symbolic logic},
  volume={44},
  number={1},
  pages={36--50},
  year={1979},
  publisher={Cambridge University Press}
}

@book{Ono,
  title={Residuated lattices: an algebraic glimpse at substructural logics},
  author={Galatos, Nikolaos and Jipsen, Peter and Kowalski, Tomasz and Ono, Hiroakira},
  volume={151},
  year={2007},
  publisher={Elsevier}
}

@article{Avron,
  title={The semantics and proof theory of linear logic},
  author={Avron, Arnon},
  journal={Theoretical Computer Science},
  volume={57},
  number={2-3},
  pages={161--184},
  year={1988},
  publisher={Elsevier}
}

@book{Troelstra,
author = {Troelstra, Anne Sjerp},
address = {Stanford},
booktitle = {Lectures on linear logic},
isbn = {0937073784},
lccn = {lc 91038902},
publisher = {Center for the Study of Language and Information},
series = {CSLI lecture notes; no.29},
title = {Lectures on linear logic },
year = {1992}
}

@article{Hrubes,
title = {On lengths of proofs in non-classical logics},
journal = {Annals of Pure and Applied Logic},
volume = {157},
number = {2},
pages = {194-205},
year = {2009},
note = {Kurt Gödel Centenary Research Prize Fellowships},
issn = {0168-0072},
doi = {https://doi.org/10.1016/j.apal.2008.09.013},
url = {https://www.sciencedirect.com/science/article/pii/S0168007208001292},
author = {Pavel Hrubeš}
}

@book{Krajicek,
  address = {New York, NY},
  author = {Jan Kraj\'{i}{\v{c}}ek},
  editor = {},
  publisher = {Cambridge University Press},
  title = {Proof Complexity},
  year = {2019},
  series = {Encyclopedia of Mathematics and its Applications},
  volume = {170},
  isbn = {9781108416849}
}

@article{AmirProofComp,
  title={Proof Complexity and Feasible Interpolation},
  author={Akbar Tabatabai, Amirhossein},
  journal={arXiv preprint arXiv:2505.03002},
  year={2025}
}

@article{Pudlak,
  title={On the complexity of intuitionistic propositional calculus},
  author={Pudl{\'a}k, Pavel},
  journal={Sets and Proofs},
  volume={258},
  pages={197},
  year={1999},
  publisher={Cambridge University Press}
}

@article{PudlakBuss,
  title={On the computational content of intuitionistic propositional proofs},
  author={Buss, Samuel R and Pudl{\'a}k, Pavel},
  journal={Annals of Pure and Applied Logic},
  volume={109},
  number={1-2},
  pages={49--64},
  year={2001},
  publisher={Elsevier}
}

@article{BussMints,
  title={The complexity of the disjunction and existential properties in intuitionistic logic},
  author={Buss, Samuel  R and Mints, Grigori},
  journal={Annals of Pure and Applied Logic},
  volume={99},
  number={1-3},
  pages={93--104},
  year={1999},
  publisher={Elsevier}
}

@article{krajivcekFeasible,
  title={Lower bounds to the size of constant-depth propositional proofs},
  author={Kraj{\'\i}{\v{c}}ek, Jan},
  journal={The Journal of Symbolic Logic},
  volume={59},
  number={1},
  pages={73--86},
  year={1994},
  publisher={Cambridge University Press}
}

@article{krajivcekfeasible2,
  title={Interpolation theorems, lower bounds for proof systems, and independence results for bounded arithmetic},
  author={Kraj{\'\i}{\v{c}}ek, Jan},
  journal={The Journal of Symbolic Logic},
  volume={62},
  number={2},
  pages={457--486},
  year={1997},
  publisher={Cambridge University Press}
}

@article{Hrubes1,
  title={A lower bound for intuitionistic logic},
  author={Hrube{\v{s}}, Pavel},
  journal={Annals of Pure and Applied Logic},
  volume={146},
  number={1},
  pages={72--90},
  year={2007},
  publisher={Elsevier}
}

@article{Hrubes2,
  title={Lower bounds for modal logics},
  author={Hrube{\v{s}}, Pavel},
  journal={The Journal of Symbolic Logic},
  volume={72},
  number={3},
  pages={941--958},
  year={2007},
  publisher={Cambridge University Press}
}

@article{Haken,
  title={The intractability of resolution},
  author={Haken, Armin},
  journal={Theoretical computer science},
  volume={39},
  pages={297--308},
  year={1985},
  publisher={Elsevier}
}

@article{Cutting2,
  title={Lower bounds for resolution and cutting plane proofs and monotone computations},
  author={Pudl{\'a}k, Pavel},
  journal={The Journal of Symbolic Logic},
  volume={62},
  number={3},
  pages={981--998},
  year={1997},
  publisher={Cambridge University Press}
}

@article{tabatabai2025,
  title={Universal proof theory: Feasible admissibility in intuitionistic modal logics},
  author={Akbar Tabatabai, Amirhossein and Jalali, Raheleh},
  journal={Annals of Pure and Applied Logic},
  volume={176},
  number={2},
  pages={103526},
  year={2025},
  publisher={Elsevier}
}

@inproceedings{ciabattoni2008axioms,
  title={From axioms to analytic rules in nonclassical logics},
  author={Ciabattoni, Agata and Galatos, Nikolaos and Terui, Kazushige},
  booktitle={2008 23rd Annual IEEE Symposium on Logic in Computer Science},
  pages={229--240},
  year={2008},
  organization={IEEE}
}

@article{ciabattoni2012algebraic,
  title={Algebraic proof theory for substructural logics: cut-elimination and completions},
  author={Ciabattoni, Agata and Galatos, Nikolaos and Terui, Kazushige},
  journal={Annals of Pure and Applied Logic},
  volume={163},
  number={3},
  pages={266--290},
  year={2012},
  publisher={Elsevier}
}

@article{tabatabai2026proof,
  title={Proof Complexity of Linear Logics},
  author={Akbar Tabatabai, Amirhossein and Jalali, Raheleh},
  journal={arXiv preprint arXiv:2601.22393},
  year={2026}
}
\end{document}